\documentclass[fontsize=12pt,a4paper,headings=normal,
twoside=false,leqno,parskip=half-,abstract=true]{scrartcl}
\usepackage{amsmath,amsthm,amssymb}
\usepackage{comment}
\usepackage{xr}
\usepackage{xcolor}
\usepackage{mdframed}
\usepackage{blkarray}
\usepackage{graphicx}
\usepackage{appendix}
\usepackage{multicol}
\newtheorem{theorem}{Theorem}%[section]
\newtheorem{lemma}[theorem]{Lemma}

\newtheorem{corollary}[theorem]{Corollary}

\newtheorem{prop}[theorem]{Proposition}
\newtheorem{conjecture}[theorem]{Conjecture}
\theoremstyle{definition}
\newtheorem{definition}[theorem]{Definition}
\newtheorem{obs}[theorem]{Observation}
\DeclareMathOperator{\sign}{sign}
\DeclareMathOperator{\rk}{rank}
\DeclareMathOperator{\tr}{tr}
\newcommand{\coloneqq}{\mathrel{:=}} 
\usepackage{hyperref}
\hypersetup{
 pdftitle={},
 pdfauthor={},
 colorlinks=true,
 linkcolor=blue,
 citecolor=blue,
 filecolor=blue,
 urlcolor=blue}

\usepackage[english]{babel}

\DeclareMathAlphabet{\mathpzc}{OT1}{pzc}{m}{it}
\DeclareMathOperator{\sgn}{sign}
\let\Re\undefined
\DeclareMathOperator{\Re}{Re}

\begin{document}
\title{Unstable Cores are the source of instability in chemical
    reaction networks}
\author{ \\
{~}\\Nicola Vassena$\;^{a}$, Peter F.\ Stadler$\;^{a,\,b,\,c,\,d,\,e}$}
\vspace{2cm}

\date{\today}
\maketitle
\thispagestyle{empty}

\begin{abstract}
 In biochemical networks, complex dynamical features such as superlinear
 growth and oscillations are classically considered a consequence of
 autocatalysis. For the large class of parameter-rich kinetic models, which
 includes Generalized Mass Action kinetics and Michaelis-Menten kinetics,
 we show that certain submatrices of the stoichiometric matrix, so-called
 unstable cores, are sufficient for a reaction network to admit instability
 and potentially give rise to such complex dynamical behavior. The
 determinant of the submatrix distinguishes unstable-positive feedbacks,
 with a single real-positive eigenvalue, and unstable-negative feedbacks
 without real-positive eigenvalues.  Autocatalytic cores turn out to be
 exactly the unstable-positive feedbacks that are Metzler matrices. Thus
 there are sources of dynamical instability in chemical networks that are
 unrelated to autocatalysis. We use such intuition to design
 non-autocatalytic biochemical networks with superlinear growth and
 oscillations.\\
  
\\ 
\end{abstract}

\vfill

 \par\noindent\rule{\textwidth}{0.5pt}
 \par

\footnotesize{[a] Interdisciplinary Center for Bioinformatics, Universit{\"a}t
  Leipzig, Germany.
  [b] Bioinformatics Groups, Department of Computer Science; Competence
  Center for Scalable Data Services and Solutions Dresden/Leipzig
  (scaDS.AI); German Centre for Integrative Biodiversity Research (iDiv);
  Leipzig Research Center for Civilization Diseases; School of Embedded and
  Compositive Artificial Intelligence, Universit{\"a}t Leipzig, Germany. [c] Department of Theoretical Chemistry, University of Vienna,
Austria. [d] Facultad de Ciencias, Universidad National de Colombia, Sede
  Bogot{\'a}, Colombia. [e] Santa Fe Institute, Santa Fe, New Mexico}

\newpage

\tableofcontents

\newpage

%%%%%%%%%%%%%%%%%%%%%%%%%%%

%%%%%%%%%% Insert the texts which can accomdate on firstpage in the tag "fmtext" %%%%%

\section{Introduction}
%%%% Insert A head here
Chemical Reaction Network theory has striven to elucidate the connection
between the structure of a (bio)chemical reaction network and its dynamical
behavior since Rutherford Aris' seminal work \cite{Aris:64}.  A question of
particular interest for applications in biological systems is whether, for
a given reaction network, there is a choice of kinetic parameters that
renders an equilibrium dynamically unstable and thus may give rise to
complex dynamical behavior.

Well-known examples of such instabilities, such as the famous
Belusouv--Zhabotinsky reactions \cite{Zhabotinsky:64,Cassani:21}, typically
involve autocatalysis, i.e., the presence of a species that catalyzes its
own production.  The concept of autocatalysis introduced by Wilhelm Ostwald
\cite{Ostwald:1890} in the context of chemical kinetics refers to a
temporary speed-up of the reaction before it settles down to reach
equilibrium, see e.g.\ \cite{Bissette:13,Schuster:19}.  Autocatalysis plays
a key role both in extant metabolic networks and in models of the origin of
life. Autocatalytic pathways, such as glycolysis, contain reactions that
consume some of the pathway's products and exhibit positive feedback
\cite{Kun:08,Siami:20}. In order to explain the emergence of
self-replicating organisms, ``collectively autocatalytic'' networks of
interacting molecules have been proposed \cite{Kauffman:86}. In the setting
of metabolic reaction networks, the importance of ``network autocatalysis''
has been emphasized \cite{Smith:04,Smith:13}.  In a more general setting,
such structures have been studied as self-maintaining chemical
organizations \cite{Kaleta:06,Benkoe:09a}. Despite the central role of
autocatalysis, however, its formal, mathematical understanding is
limited. Different concepts of autocatalysis have been proposed
\cite{Andersen:12a,Deshpande:14,Barenholz:17,blokhuis20, unterberger22};
for a comparison see \cite{DefAut20}.

Even though autocatalysis has played a central role in investigations of
complex dynamics, including deterministic chaos \cite{Gyorgyi:91},
dynamical systems theory does not seem to imply that autocatalysis has to
be necessarily invoked to explain the existence of dynamic instabilities.
In fact, the unstable manifold theorem \cite{Hsubook} suggests superlinear
divergence from an unstable equilibrium, while in a supercritical Hopf
bifurcation \cite{GuHo84,Hsubook} a stable equilibrium loses stability and
generates a stable periodic orbit. Neither of these situations requires
assumptions that even vaguely resemble autocatalysis. In both cases, the
key property is the loss of stability of an equilibrium. This begs the
question whether the existence of unstable equilibria can be explained in
terms of the structure of the chemical reaction network, or more precisely,
in terms of its stoichiometric matrix. Since autocatalytic cores are also
characterized in terms of the stoichiometric matrix \cite{blokhuis20}, we
set out here to disentangle instability and autocatalysis in terms of a
purely structural view on chemical reaction networks. On the one hand, we
confirm that autocatalysis has always the potential to destabilize an
equilibrium. Autocatalysis is thus a sufficient condition for instability
in certain parameter regions. On the other hand, we establish that
autocatalysis is not necessary for superlinear growth or oscillations.

In more detail, we shall see that small unstable subsystems are sufficient
to imply that an entire reaction network admits instability in kinetic
models that are sufficiently `rich' in parameters (Def.~\ref{def:p-rich}),
such as Michaelis-Menten \cite{MM13}, Hill \cite{Hill10}, and Generalized
Mass Action \cite{Muller:12} kinetics. A key ingredient towards this result
are the so-called Child-Selections (CS). Conceptually, a CS $\pmb{\kappa}$
is a square submatrix of the stoichiometric matrix $S$ comprising $k$
species and $k$ reactions such that there is a 1-1 association between
species $m$ and reactions $j(m)$ for which $m$ is a reactant for
$j(m)$. From this submatrix, we obtain the \emph{Child-Selection matrix}
$S[\pmb{\kappa}]$ by reordering the columns so that species and their
associated reactions correspond to the diagonal entries. The existence of
an unstable Child-Selection matrix is sufficient for a network
$\pmb{\Gamma}$ with a parameter-rich kinetic model to admit instability
(Cor.~\ref{CSinstability}).  Since Child-Selections can be
``concatenated'', minimal unstable Child-Selections are well defined, in
the sense that the Child-Selection matrix $S[\pmb{\kappa}]$ does not
contain again a proper principal submatrix that is an unstable
Child-Selection matrix, which leads to the main definition of this paper
(Def. \ref{def:unstablecore}):
\begin{center}
\emph{An \emph{unstable core} is an Hurwitz-unstable Child-Selection matrix
for which no principal submatrix is Hurwitz-unstable.}
\end{center}

Unstable cores come in two flavors: An unstable core is an
\emph{unstable-negative feedback} if $\displaystyle \sign\det
S[\pmb{\kappa}]=(-1)^{k}$ and \emph{unstable-positive feedback} if
$\displaystyle \sign\det S[\pmb{\kappa}]=(-1)^{k-1}$.  This terminology is
inspired by positive and negative feedback cycles, where the sign of the
feedback is typically defined as the sign of the product of the
off-diagonal entries. The condition on the determinants translates into a
structural difference in their spectra: Unstable-positive feedbacks have a
single real-positive eigenvalue, while unstable-negative feedbacks have no
real-positive eigenvalues at all (Lemma~\ref{lem:feedback}).  The presence
of related ``positive feedbacks'' is known to be \emph{necessary} for
multistationarity \cite{Soule2003,Ba-07} (see also \textbf{EXAMPLE
  A}). Oscillations, on the other hand, can be induced both by
unstable-positive and unstable-negative feedbacks (see \textbf{EXAMPLE B}
and \textbf{EXAMPLE C}).

We then shall turn our attention to autocatalysis and autocatalytic cores
in particular, and formally define an \emph{autocatalytic matrix}
(Def.~\ref{def:autocatalysis}), inspired by
\cite{blokhuis20}. Autocatalytic cores are then autocatalytic submatrices
of the stoichiometric matrix that do not contain any autocatalytic
submatrix. Our main result (Thm.~\ref{thm:autocoremain}) characterizes
autocatalytic cores in terms of unstable cores and Metzler matrices,
i.e., matrices with nonnegative off-diagonal entries, and can be expressed
as follows:
\begin{center}
  \emph{An autocatalytic core is an unstable-positive feedback that in
  addition is a Metzler matrix.}
\end{center}
Recalling that autocatalytic networks must contain an autocatalytic core
leads to a convenient characterization of general autocatalytic networks
(Cor.~\ref{cor:autocatnetwork}):
\begin{center}
  \emph{A network is autocatalytic if and only if there is a
  Child-Selection matrix that is a Hurwitz-unstable Metzler matrix.}
\end{center}
Our results show in particular that \emph{every autocatalytic network
admits instability}. The converse clearly is not true, i.e.,
\emph{autocatalysis is not necessary for instability}: First,
unstable-negative feedbacks are never autocatalytic
(Cor.~\ref{cor:noautocat} and \textbf{EXAMPLE C}). Second, it is
straightforward to construct unstable-positive feedbacks that are not
Metzler matrices and thus not autocatalytic (see \textbf{EXAMPLES A, B}, and
\textbf{EXAMPLE D} where we list non-autocatalytic unstable-positive
feedbacks in the sequential and distributive double phosphorylation).
  
The paper is organized as follows. Section \ref{sec:crn} presents the
general setting of reaction networks and Section \ref{sec:parameterrich}
the definition of parameter-rich models. Section \ref{sec:CS} introduces
Child-Selections and employs a Cauchy--Binet analysis to expand each
coefficient of the characteristic polynomial by Child-Selections: This is used to connect the topology of the network to stability properties.  Section
\ref{sec:unstablecores} introduces the
definition of \emph{unstable cores} and its generalization, \emph{D-unstable cores}. Section \ref{sec:feedbacks} addresses
unstable-positive and unstable-negative feedbacks. In Section
\ref{sec:autocatMetzler} we focus on autocatalytic cores and we show in
Theorem \ref{thm:autocoremain} that they are a special case of unstable
cores. Section \ref{sec:example} presents four examples of
non-autocatalytic networks that shows either multistationarity (with
superlinear growth) or oscillations. We refer to the Supplementary Material
(SM) for a full analysis and explanation of such examples. Albeit our paper shows that autocatalysis is not necessary for
complex behavior as oscillations or superlinear growth, in Section
\ref{sec:whatisspecial} we nevertheless provide an informal indication
why autocatalysis has been the typical source of instability in well-known
examples. Section \ref{sec:conclusion} concludes the paper, discussing the
results in a general context. {The technical proof of Theorem \ref{thm:autocoremain} is presented in Section \ref{sec:NC_Proof}.}

\section{Chemical reaction networks, dynamics, and stability}\label{sec:crn}

A chemical reaction network $\pmb{\Gamma}$ is a set $M$ of
\emph{species} together with a set $E$ of \emph{reactions}. A
reaction $j\in E$ is a pair of formal linear combinations
\begin{equation}\label{reaction}
  s^j_{m_1}m_1+...+s^j_{m_{|M|}}m_{|M|} \quad
  \underset{j}{\rightarrow}\quad
  \tilde{s}^j_{m_1}m_1+...+\tilde{s}^j_{m_{|M|}}m_{|M|},
\end{equation}
with \emph{stoichiometric} coefficients $s^j_m\ge0$ and $\tilde{s}^j_m\ge
0$. Usually, one assumes $s^j_m,\tilde{s}^j_m\in\mathbb{N}_0$, although
this restriction is not relevant here. A species $m$ is a \emph{reactant}
of $j$ if $s^j_m>0$ and a \emph{product} of $j$ if $\tilde{s}^j_m>0$.  Note
that \eqref{reaction} assigns a fixed order to any reaction $j$, so that we
treat all reactions as \emph{irreversible}. However, reversible processes
such as $j:\, A \rightleftharpoons B$ can be naturally taken into account
by considering two opposite reactions $j_1:\, A\to B$ and $j_2:\, B\to
A$. In biological systems, reaction networks are often \emph{open}: They
exchange chemicals with the outside environment. For this reason, we also
consider reactions with no outputs (\emph{outflow reactions}) or with no
inputs (\emph{inflow reactions}). These describe the exchange of material
between the system and its environment.

A reaction $j_{C}$ is \emph{explicitly catalytic} if a species $m$ is both
a reactant and a product of the reaction $j_C$, i.e., $s^{j_{C}}_m
\tilde{s}^{j_{C}}_m \neq 0$. The net change of chemical composition in a
reaction network is described by the $|M|\times |E|$ \emph{stoichiometric
matrix} $S$ defined as
\begin{equation}\label{S}
    S_{mj}:=\tilde{s}^j_m - s^j_m,
\end{equation}
Note that for explicit catalysts $m$ with $s^{j_{C}}_m=
\tilde{s}^{j_{C}}_m$, the stoichiometric coefficients cancel and thus do
not appear in the entry $S_{m{j_C}}$ of the stoichiometric matrix.

The time-evolution $x(t)\ge0$ of an $|M|$-vector of the concentrations
under the assumption of spatial homogeneity, e.g.\ in a well-mixed reactor,
obeys the system of ordinary differential equations
\begin{equation}\label{maineq}
\dot{x}=f(x) \coloneqq Sr(x),
\end{equation}
where $r: \mathbb{R}^{|M|}_{\ge0} \to \mathbb{R}^{|E|}_{\ge0}$ is a vector
of reaction rate functions. In this contribution, we will in particular be
concerned with \emph{equilibria} or \emph{fixed points}, i.e.,
$|M|$-vectors $\bar{x}$ satisfying $f(\bar{x})=0$.
{Throughout, we assume that the stoichiometric matrix $S$ admits a strictly positive kernel vector $\mathbf{r}\in \mathbb{R}^{|E|}_{\ge0}$, i.e. $S\mathbf{r}=0$, which is in turn a basic necessary condition for the existence of any equilibrium in the first place. Networks satisfying this standard assumption have been called \emph{consistent} in the literature, see for example \cite{Ang07}.}
Assuming that $f$ is a
continuously differentiable vector field, the stability of an equilibrium
$\bar{x}$ is determined by the Jacobian matrix $G$ evaluated at $\bar{x}$,
which has entries $G_{hm}(\bar{x})\coloneqq \partial f_h(x)/\partial
x_m|_{x=\bar{x}}$. A real square matrix $A$ is \emph{Hurwitz-stable} if all
its eigenvalues $\lambda$ satisfy $\Re\lambda<0$. It is
\emph{Hurwitz-unstable} if there is at least one eigenvalue $\lambda_u$
with $\Re\lambda_u>0$. It is well known that an equilibrium $\bar{x}$ is
dynamically stable if $G(\bar{x})$ is Hurwitz-stable, and unstable if
$G(\bar{x})$ is Hurwitz unstable \cite{Hsubook}.

We consider all reactions to be directional. Thus the rate $r_j$ of a
reaction $j$ is a non-negative function. We further assume that it depends
only on the concentrations of its reactants, i.e., the molecular species
$m$ with stoichiometric coefficients $s^{j}_{m}>0$.  Finally, a reaction
can only take place if all its reactants are present, in which case the
rate increases with increasing reactant concentrations. We capture these
conditions in the following formal definition of a kinetic model for a
given chemical reaction network:
\begin{definition}\label{def:kineticmodel}
  Let $\pmb{\Gamma}=(M,E)$ be a chemical reaction network. A
  differentiable function $r:\mathbb{R}^M_{\ge0}\to\mathbb{R}^E$ is a
  \emph{kinetic model} for $\pmb{\Gamma}$ if 
  \begin{enumerate}
  \item $r_j(x)\ge0$ for all
  $x$, 
  \item $r_j(x)>0$ implies $x_m>0$ for all $m$ with $s^j_{m}>0$,
\item  $s^j_{m}=0$ implies $\partial r_j/\partial x_m\equiv 0$,
\item if
  $x>0$ and $s^j_{m}>0$ then $\partial r_j/\partial x_m>0$.
  \end{enumerate}
\end{definition}
For $\bar{x}\in\mathbb{R}^M_{\ge 0}$ we write
\begin{equation} 
  r'_{jm}(\bar{x}):= \frac{\partial r_j(x)}{\partial x_m}
  \bigg\vert_{x=\bar{x}} .
\end{equation}
The $|E|\times |M|$ matrix $ R(\bar{x})$ with entries
$R_{jm}:=r'_{jm}(\bar{x})$ is called the \emph{reactivity matrix}. By
construction, we have $r'_{jm}(\bar{x})\ge 0$ for all $m\in M$ and $j\in
E$, i.e., $R(\bar{x})$ is a non-negative matrix. For $\bar{x}>0$, moreover,
$r'_{jm}(\bar{x})>0$ if and only if $s^j_{m}>0$. Thus, for any strictly
positive equilibrium $\bar{x}>0$, the signs of $R(\bar{x})$ are completely
determined by the reactants in the network. Together, the stochiometric
matrix $S$ and the reactivity matrix $R$ determine the stability of an
equilibrium, since \eqref{maineq} implies that the Jacobian $G$ is of the
form
\begin{equation}
  G(\bar{x}) = S R(\bar{x}).
\end{equation}
Prominent examples of kinetic models include \emph{Mass Action kinetics}
\cite{MA67}
\begin{equation}\label{MAeq}
  r_j(x) := a_j\prod_{m\in M} x_m^{s_{m}^{j}} 
\end{equation}
and \emph{Michaelis--Menten kinetics} \cite{MM13}
\begin{equation}\label{MMeq}
  r_j(x) : =a_j\prod_{m\in M} \Bigg( \frac{x_m}{(1+b^j_m
    x_m)}\Bigg)^{s^j_m}.
\end{equation}
Note that Mass Action kinetics appears as the limiting case of
Michaelis--Menten kinetics with $b^j_m=0$ for all $j,m$. Here, $r_j(x)$ depends on parameters ($a_j$ and $b_m^j$). In general, we
write $r(x;p)$ for a parametric kinetic model that depends on parameters
$p$.  In this contribution we are not interested in concrete choices of
such parameters. Instead, we are interested in the existence of unstable
equilibria given a suitable choice of parameters.
\begin{definition}
  A network $\pmb{\Gamma}=(M,E)$ with a parametrized kinetic model $r(x;p)$
  \emph{admits instability} if there exists a choice $\bar{p}$ of
  parameters such that there is a positive equilibrium $\bar{x}$ of
  $\dot{x}=S r(x;\bar{p})$ with a Hurwitz-unstable Jacobian $G(\bar{x})$.
  \label{def:admitinstabil}
\end{definition}

\section{Parameter-rich kinetic models} \label{sec:parameterrich}
We are particularly interested in kinetic models $r(x;p)$ that have a
sufficient number of free parameters $p$ such that the equilibrium
$\bar{x}$ and $R(\bar{x})$ can be chosen independently.  We formalize this
idea as follows:
\begin{definition}
  A kinetic rate model $r(x;p)$ is \emph{parameter-rich} if, for every
  positive equilibrium $\bar x>0$ and every $|E|\times |M|$ matrix $R$ with
  entries satisfying $r'_{jm}>0$ if $s_{m}^{j}>0$ and $r'_{jm}=0$ if
  $s_{m}^{j}=0$, there are parameters $\bar{p}=p(\bar{x},R)$ such that
  $\frac{\partial r_j(x;\bar{p})}{\partial x_m}\big|_{x=\bar{x}} = r'_{jm}$.
  \label{def:p-rich}
\end{definition}
{For any choice of a positive vector $\bar{x}$ and  matrix
  entries $r'_{jm}$ with $s_{m}^{j}>0$, in parameter-rich
  models it is possible to find parameter values $\bar{p}$ such that $\bar{x}$ is a fixed
  point and the $r'_{jm}$ are the partial derivatives at the equilibrium $\bar{x}$. The only constraints on the Jacobian, therefore, derives from the stoichiometry of the network.
  }  In the following
we write $\pmb{r'}$ for an arbitrary non-zero choice of the $r'_{jm}$ with
$s_{m}^{j}>0$. It is then convenient to think of the partial derivatives
$r'_{jm}$ themselves as \emph{symbols} that can be specialized to
particular positive values at our convenience.  We write $R(\pmb{r'})$ for
the corresponding symbolic \emph{reactivity matrix}. The Jacobian
$G(\pmb{r'}) = S R(\pmb{r'})$ thus can also be viewed as a symbolic
matrix. Combining this with Def.~\ref{def:admitinstabil} and
Def.~\ref{def:p-rich} we immediately arrive at the following
\begin{obs}
  \label{obs:instability}
  A network $\pmb{\Gamma}$ with a parameter-rich kinetic model admits
  instability if and only if there is a choice of symbols $\pmb{r}'$
  such that the symbolic Jacobian $G(\pmb{r}')$ is Hurwitz-unstable.
\end{obs} 

{ Def.~\ref{def:p-rich} requires parameter-rich
  kinetic models to have enough parameters to simultaneously satisfy
  constraints at two levels: at the level of the functions
  $r(x,p)$ and at the level of their first derivatives. The former comes
  from the equilibrium constraints, i.e. $S\pmb{r}(x)=0$, while the latter
  comes from the matrix $R$ that -- jointly with the stoichiometric matrix
  $S$ -- prescribes the Jacobian. This fact suggests that at least two
  parameters for each reaction rate $r_j$ must be present in order to
  account for both levels of constraints. Mass action kinetics,
 Eq.\eqref{MAeq}, presents one single parameter for each reaction
  rate, while Michaelis--Menten kinetics, Eq.\eqref{MMeq}, at least
  two. This observation indicates that mass action kinetics is not
  parameter-rich, while Michaelis--Menten {is}. Theorem 6.1 in
  \cite{VasHunt} confirms this expectation regarding Michaelis--Menten
  kinetics:}
\begin{lemma}\label{lemma:mmrich}
  Michaelis--Menten kinetics, Eq.(\ref{MMeq}), is parameter-rich.
\end{lemma}
In the SM, Section 1.1, we include a short proof of Lemma
\ref{lemma:mmrich} to make this contribution self-contained. Clearly, any
kinetic model that contains Michaelis--Menten kinetics as a special case is
also parameter-rich, e.g.\ Hill kinetics \cite{Hill10}.  In contrast,
we formally confirm that mass action is \emph{not}
parameter-rich: The derivative of the reaction rate of $j$ with respect to
the concentration $x_m$ of one of its reactants $m$ reads
\begin{equation}\label{eq:maderivative}
  r'_{jm}(x)=\;{s^j_m}\; x_m^{(s^j_m - 1)} \;
  k_j \prod_{n \neq m}  x_n^{s^j_n}=\frac{s^j_m}{x_m} r_j(x).
\end{equation}
For a fixed equilibrium $\bar{x}>0$, the relation
$r'_{jm}(\bar{x})=\frac{s^j_m}{\bar{x}_{m}}r_j(\bar{x})$ shows the absence
of parameter freedom to harness the value of the derivatives
$r'_{jm}(\bar{x})$ independently from the values of the fluxes
$r_j(\bar{x})$ and the concentration $\bar{x}$. Thus, if the
fluxes $ \bar{r}_j=r_j(\bar{x})$ of a fixed concentration $\bar{x}$ solve
the equilibrium constraints $S \bar{r}=0$, the values of the derivatives
$r'_{jm}(\bar{x})$ cannot be chosen independently and with freedom, contradicting Def.~\ref{def:p-rich}. {In particular, the set of Jacobian matrices of an equilibrium of a mass-action system on a network $\pmb{\Gamma}$ is a \emph{subset} of the set of Jacobian matrices of an equilibrium of any parameter-rich models on $\pmb{\Gamma}$. This has two important consequences: The fact that $\pmb{\Gamma}$ endowed with parameter-rich kinetics admits instability does not imply that the same network $\pmb{\Gamma}$ with mass-action kinetics admits instability. In contrast, the fact that $\pmb{\Gamma}$ with parameter-rich kinetics does not admit instability implies that $\pmb{\Gamma}$ does not admit instability with mass-action kinetics as well.}

On the other hand, \emph{Generalized Mass Action kinetics} considers the
stoichiometric exponent $s^j_m$ appearing in mass action to be an {additional} positive
parameter itself. Such kinetics has been proposed to model
intra-cellular reactions where the assumption of spatial homogeneity fits
less. See \cite{Muller:12} for more details. Let $c^j_m\in\mathbb{R}_{\ge
  0}$ be the real parametric exponent of the species $m$ in the reaction
$j$ and assume that $c^j_m > 0$ if and only if $s^j_m>0$, and zero
otherwise. Generalized Mass Action kinetics reads
\begin{equation}\label{eq:gma}
  r_j(x) := a_j\prod_{m\in M} x_m^{c_{m}^{j}}, \quad \quad \quad
  \text{with $c^j_m \neq 0$ if and only if $s^j_m\neq 0$}.
\end{equation}
Analogously to Lemma~\ref{lemma:mmrich}, we have the following.
\begin{lemma}
Generalized mass action, \eqref{eq:gma}, is parameter-rich.
\end{lemma}
\begin{proof}
The same computation as \eqref{eq:maderivative} leads to
\begin{equation}
   r'_{jm}(x)=\frac{c^j_m}{x_m} r_j(x), 
\end{equation}
where now $c^j_m$ is not a fixed integer value but a real-valued
parameter. Thus, as in the proof of Lemma~\ref{lemma:mmrich}, we can fix
any concentration $\bar{x}>0$, equilibrium fluxes
$\bar{r}\in\mathbb{R}^{|E|}_{>0}$ with $S\bar{r}=0$ and any choice of
symbols $\bar{\pmb{r}}'$. For the positive choice of parameters
\begin{equation}
\begin{cases}
\mathpzc{c}^j_m:=\bar{x}_m \bar{r}'_{jm}(\bar{r}_j)^{-1},\\
\mathpzc{a}_j:=\bar{r}_j \big(\prod_{m\in M} \bar{x}_m\big)^{-\bar{x}_m \bar{r}'_{jm}(\bar{r}_j)^{-1}}
\end{cases}
\end{equation}
the generalized mass-action functions $r_j(x):=\mathpzc{a}_j\prod_{m\in M}
x_m^{\mathpzc{c}_{m}^{j}}$ satisfy $Sr(\bar{x})=0$, i.e. $\bar{x}$ is an
equilibrium, with prescribed partial derivatives $\bar{\pmb{r}}'$. Thus
Generalized Mass Action kinetics is parameter-rich.
\end{proof}

{Finally, we point out that some parameter-rich kinetics can themselves be derived as a singular limit of mass-action systems: This should not create confusion. To clarify this in an explicit example, consider Michaelis--Menten kinetics. It is well known that a reaction $A \rightarrow B$ endowed with Michaelis-Menten kinetics can be seen as the singular limit of the three-reactions mass-action network $A+E\rightleftharpoons I \rightarrow B+E$. In this latter network, however, three reactions instead of one appear. Consequently, the associated mass-action system presents three parameters, not one. In contrast, the very network $A\rightarrow B$ presents only one parameter when endowed with mass action. In this viewpoint, the parameter-richness of Michaelis--Menten is inherited by such expanded slow-fast mechanism under mass-action constraints.  Def.~\ref{def:p-rich} treats nevertheless the structure of the network as given and fixed, and it does not consider expanded 'elementary' versions of the same network.}

\section{Child-Selections, CB-summands, elementary CB-components}\label{sec:CS}
A central tool in this work are bijective associations between molecular
species and reactions that give rise to square submatrices of $S$.
\begin{definition}
  Let $\pmb{\Gamma}=(M,E)$ be a network with stoichiometric matrix $S$.
  A \emph{$k$-Child-Selection triple}, or $k$-CS for short, is a triple
  $\pmb{\kappa}=(\kappa,E_{\kappa},J)$ such that $|\kappa|=|E_{\kappa}|=k$,
  $\kappa\subseteq M$, $E_{\kappa}\subseteq E$, and $J:\kappa\to
  E_{\kappa}$ is a bijection satisfying $s_m^{J(m)}>0$ for all
  $m\in\kappa$. We call $J$ a \emph{Child-Selection bijection}.
\end{definition}
Note that since $J$ is a bijection between two ordered sets, we can
naturally consider the \emph{signature} (or parity) $\sgn{J}$ of the map
$J$, where $J$ is seen as a permutation of a set of cardinality $k$. For an
$|M|\times|E|$ matrix $A$, and subsets $\kappa \subseteq M$, $\iota
\subseteq E$, the notation $A[\kappa,\iota]$ refers to the submatrix of $A$
with rows in $\kappa$ and columns in $\iota$. For square matrices,
principal submatrices with $\kappa=\iota$ are indicated as $A[\kappa]$. A
$k$-CS $\pmb{\kappa}$ identifies a $k\times k$ submatrix
$S[\kappa,E_{\kappa}]$ of $S$. Its columns can be reordered such that the
reactions $J(m)$ appear in the same order as their corresponding species
$m$. This gives rise to a matrix $S[\pmb{\kappa}]$ with entries
\begin{equation}\label{def:csmatrixentries}
S[\pmb{\kappa}]_{ml} \coloneqq S[\kappa,E_{\kappa}]_{m,J(l)} =
  \tilde s_{m}^{J(l)} - s_{m}^{J(l)},
\end{equation}
where the permutation of the columns of $S[\pmb{\kappa}]$ from
$S[\kappa,E_\kappa]$ is described by the Child-Selection bijection $J$ with
signature $\sgn J$. In particular,
\begin{equation}\label{eq:jperm}
  \det S[\pmb{\kappa}]=
  \sgn{J}\det S[\kappa,E_{\kappa}],
\end{equation}
since the determinant is a multilinear and alternating form of the columns
of a matrix. We call a matrix $S[\pmb{\kappa}]$ arising from a $k$-CS
$\pmb{\kappa}$ in $\pmb{\Gamma}$ a \emph{Child-Selection matrix}.

Let $\pmb{\kappa}=(\kappa,E_{\kappa},J)$ be a $k$-CS.  Consider
$\kappa'\subseteq\kappa$, $E_{\kappa'}=\{J(m)\;|\;m\in\kappa'\}$ and
$J':\kappa'\to E_{\kappa'}$ with $J'(m)=J(m)$. Clearly,
$E_{\kappa'}\subseteq E_{\kappa}$ and $J':\kappa'\to E_{\kappa'}$ is the
restriction of $J$ to $\kappa'$, and thus bijective. In particular
$\pmb{\kappa'}=(\kappa',E_{\kappa'},J')$ is a $|\kappa'|$-CS. We say that
$\pmb{\kappa'}$ is a \emph{restriction} of $\pmb{\kappa}$. {Child-Selections come thus with a natural notion of minimality
w.r.t.\ some property $\mathbb{P}$.
\begin{definition}
  A $k$-CS $\pmb{\kappa}$ is minimal w.r.t.\ $\mathbb{P}$ if there is no
  restriction $\pmb{\kappa'}$ of $\pmb{\kappa}$ with
  $\kappa'\subsetneq\kappa$ such that $\pmb{\kappa'}$ has property
  $\mathbb{P}$.
\end{definition}
From a matrix perspective, we have the following observation.
\begin{obs}\label{obs:minimalP}
  Let $\pmb{\kappa}$ be a $k$-CS and set $A\coloneqq
  S[\pmb{\kappa}]$. Then every principal submatrix of $A$ satisfies
$A[\kappa']=S[\pmb{\kappa'}]$, where $\pmb{\kappa'}$ is the restriction
  of $\pmb{\kappa}$ to $\kappa'\subseteq\kappa$. 
  Suppose $\mathbb{P}$ is a matrix property.
  A Child-Selection matrix $S[\pmb{\kappa}]$ is minimal w.r.t.\ $\mathbb{P}$ if no proper principal submatrix of
  $S[\pmb{\kappa}]$ has the property $\mathbb{P}$.
\end{obs}}

As addressed in \cite{VasHunt}, the $k$-CS triples and their associated
matrices provide sufficient conditions for instability in a network with
a parameter-rich kinetic model. {We only sketch here the relevant results, based on the Cauchy--Binet formula. To make this contribution self-contained, we refer to the SM Section 1.2 for detailed proofs in the present setting. First, a definition:
\begin{definition}[Cauchy--Binet summands, elementary Cauchy--Binet components]
We call the \emph{Cauchy--Binet (CB) summands} of $G$ the matrices 
\begin{equation}
G[(\kappa,E_{\kappa})]:=S[\kappa,E_\kappa]R[E_{\kappa},\kappa],
\end{equation}  
where $E_\kappa \subseteq E$ are reaction sets for which there exists at least one Child-Selection bijection $J$ with
$J(\kappa)=E_\kappa$. We call \emph{elementary
Cauchy--Binet (CB) components} of $G$ the matrices
\begin{equation}
G[\pmb{\kappa}]:=S[\pmb{\kappa}]R[\pmb{\kappa}],
\end{equation} 
where $R[\pmb{\kappa}]$ is the diagonal matrix with entries
$R[\pmb{\kappa}]_{mm}=r'_{J(m)m}$ and $J$ is the Child-Selection bijection of $\pmb{\kappa}$.
\end{definition}
Note that the CB-summands of $G$ can be seen as the symbolic Jacobian matrices of the subnetworks comprising species in $\kappa$ and reactions in $E_\kappa$. In contrast, the elementary CB-components have a more algebraic nature and lack an analogous network interpretation. For matrices that are symbolic Jacobians of networks, the
coefficients of the characteristic polynomial can be expanded along CB-summands and elementary CB-components.
 Let
\begin{equation}
    g(\lambda)=\sum_{k=0}^{|M|}(-1)^kc_k\lambda^{|M|-k}
\end{equation}
be the characteristic polynomial of the symbolic Jacobian matrix $G(
\pmb{r}')=SR(\pmb{r}')$. The coefficients $c_k$ are the sum of the
principal minors $\det{G[\kappa]}$ for all sets $|\kappa|=k$ and the following expansion holds
\begin{equation}\label{eq:completeCBexpansion}
  c_k=\sum_{(\kappa,E_{\kappa})}\det G[(\kappa,E_{\kappa})] = 
  \sum_{\pmb{\kappa}}\det G[\pmb{\kappa}],
\end{equation}
where the first sum runs over all pairs of sets $(\kappa,E_\kappa)$ with
cardinality $k$ for which there exists at least one Child-Selection
bijection $J(\kappa)=E_\kappa$. The second sum runs over all $k$-CS triples
$\pmb{\kappa}$. 
The first equality is a consequence of the Cauchy--Binet formula, while the second a consequence of the Leibniz' expansion of the determinant. The expansion \eqref{eq:completeCBexpansion} in particular shows that the characteristic
polynomial is independent of the labeling of the network.
The following lemma and corollary state respectively that $\pmb{\Gamma}$ admits
instability for parameter-rich kinetic models whenever any CB-summand or elementary CB-component
admits instability.
\begin{lemma}
  Assume that $\pmb{\Gamma}=(M,E)$ is a network with a parameter-rich
  kinetic model. Assume there exists a choice of positive symbols
  $\pmb{r}'$ such that a CB-summand $G[(\kappa,E_\kappa)]$ is
  Hurwitz-unstable.  Then the network admits instability.
  \label{lemma:CSinst} 
\end{lemma}
\begin{corollary}\label{cor:CSinst}
  Assume that $\pmb{\Gamma}=(M,E)$ is a network with a parameter-rich
  kinetic model. Assume there exists a choice of positive symbols $\pmb{r}'$
  such that an elementary CB-component $G[\pmb{\kappa}]$ is
  Hurwitz-unstable. Then there is also a choice of positive symbols such
  that the CB-summand $G[(\kappa,E_\kappa)]$ is Hurwitz-unstable and, in
  particular, the network admits instability.
\end{corollary}}

\begin{comment}
\begin{obs}
A perhaps more direct -- but equivalent -- proof of Lemma~\ref{lemma:CSinst} and Cor.~\ref{cor:CSinst} applies the same rescaling of $\pmb{r}'$ to the characteristic polynomial
  $g(\lambda,\varepsilon)$ of $G$ instead. At $\varepsilon=0$, the
  characteristic polynomial $g(\lambda,0)$ shares the roots with the
  characteristic polynomial of $G[(\kappa, E_\kappa)]$, or
  $G[\pmb{\kappa}]$ respectively, plus $|M|-k$ zero roots. Continuity of
  the roots of a polynomial with respect to its coefficients implies the
  statement for sufficiently small positive $ \varepsilon>0$.
\end{obs}
\end{comment}

\section{Unstable cores and D-unstable cores}\label{sec:unstablecores}

The results of Section \ref{sec:CS} can now be used to obtain more convenient
sufficient topological conditions for a network to admit instability. A matrix $A$ is said to be \emph{$D$-stable} if $AD$
is Hurwitz-stable for every positive diagonal matrix $D$
\cite{Arrow:58}. It is \emph{$D$-unstable} if there is a positive diagonal
matrix $D$ such that $AD$ is Hurwitz-unstable. Clearly, an elementary
CB-component is Hurwitz-unstable for some choice of parameters
if and only if $S[\pmb{\kappa}]$ is D-unstable. In particular, $D$-unstable
Child-Selection matrices $S[\pmb{\kappa}]$ are sufficient for the network
to admit instability. Unfortunately, D-stability is non-trivial to check
algorithmically \cite{Ku19}. However, choosing $D$ to be the identity
matrix immediately shows that Hurwitz-instability implies
D-instability. Thus Lemma~\ref{lemma:CSinst} in particular recovers
\begin{corollary}[Prop.~5.12 of \cite{VasHunt}]\label{CSinstability}
  Let $\pmb{\Gamma}=(M,E)$ be a network with a parameter-rich
  kinetic model and stoichiometric matrix $S$. If $\pmb{\kappa}$ is a
  $k$-CS such that $S[\pmb{\kappa}]$ is unstable then
  $\pmb{\Gamma}$ admits instability.
\end{corollary}
Since unstable Child-Selection matrices imply network instability, it is of
interest to consider minimal unstable Child-Selection matrices.
\begin{definition}\label{def:unstablecore}
  An \emph{unstable core} is an unstable Child-Selection matrix
  $S[\pmb{\kappa}]$ for which no proper principal submatrix is unstable.
\end{definition}
By Obs.~\ref{obs:minimalP}, an unstable core is thus an unstable Child-Selection
matrix that has no unstable restriction. {We recall that the eigenvalues of a reducible matrix are the union of the eigenvalues of
  its irreducible diagonal blocks. In particular, if a reducible matrix is
  Hurwitz-unstable, then one of its diagonal blocks is Hurwitz-unstable. Since a
  diagonal block is a proper principal submatrix, this yields the irreducibility of unstable cores. We state this in an observation for later use.
\begin{obs}\label{obs:irr}
An unstable core is an irreducible matrix.
\end{obs}}

Unstable cores are sufficient but not necessary causes of instability.
A natural generalization arises from considering D-instability instead of
Hurwitz-instability: 
\begin{definition}
  A \emph{D-unstable core} is a D-unstable Child-Selection matrix
  $S[\pmb{\kappa}]$ for which no proper principal submatrix is D-unstable.
\end{definition}
An immediate consequence of the definition and the fact that
Hurwitz-instability implies D-instability is that every unstable core contains a D-unstable core. Note that a D-unstable core may be Hurwitz-stable: In this case, it can be
a proper submatrix of an unstable core. It remains an open question whether or under which network conditions the
converse of Cor.~\ref{cor:CSinst} is also true:
\begin{conjecture} \label{conjecture}
  Suppose that the network does not possess any $D$-unstable core. Then
  $G(\pmb{r'})$ cannot be Hurwitz-unstable and thus the network
  $\pmb{\Gamma}$ does not admit instability.
\end{conjecture}

\section{Unstable-positive and unstable-negative feedbacks}\label{sec:feedbacks}

\begin{definition}
  Let $\pmb{\Gamma}$ be a network with a parameter-rich kinetic model.  An
  \emph{unstable-positive feedback} is an unstable core satisfying
  $\displaystyle \sign\det S[\pmb{\kappa}]=(-1)^{k-1}$; an
  \emph{unstable-negative feedback} is an unstable core satisfying
  $\displaystyle \sign\det S[\pmb{\kappa}]=(-1)^{k}$.
\end{definition}
\begin{obs}\label{obs:positivefeedbackinstability}
  If $S[\pmb{\kappa}]$ is an unstable-positive feedback then the associated
  elementary CB-component $G[\pmb{\kappa}]=S[\pmb{\kappa}]R[\pmb{\kappa}]$
  satisfies $\sign \det G[\pmb{\kappa}]=\sign \det
  S[\pmb{\kappa}]=(-1)^{k-1}$ and thus $G[\pmb{\kappa}]$ is unstable for
  any choice of symbols in $R[\pmb{\kappa}]$.
\end{obs}

As an example, consider two networks with the following stoichiometric 
matrices
\begin{equation}
  S^{+} = 
\begin{pmatrix}
    -1 & 0 & 0 & 0 & 2\\
    2 & -1 & 0 & 0 & 0\\
    0 & 2 & -1 & 0 & 0\\
    0 & 0 & 2 & -1 & 0\\
    0 & 0 & 0 & 2 & -1
\end{pmatrix}
\qquad\qquad
S^{-} = 
\begin{pmatrix} 
    -1 & 0 & 0 & 0 & -2\\
    2 & -1 & 0 & 0 & 0\\
    0 & 2 & -1 & 0 & 0\\
    0 & 0 & 2 & -1 & 0\\
    0 & 0 & 0 & 2 & -1
\end{pmatrix},
\end{equation} 
where $\det S^{+}=31$ and $\det S^{-}=-33$, corresponding to
unstable-positive and unstable-negative feedbacks, respectively. In fact:
the Hurwitz-instability of $S^+$ is clear from the sign of its determinant,
while the instability of $S^-$ can be checked by computing its eigenvalues
$(\lambda_1,\lambda_2,\lambda_3,\lambda_4,\lambda_5)\approx(-3,-1.63\pm
1.90 i, 0.62\pm 1.18 i),$ where $\lambda_4,\lambda_5$ have positive-real
part. Any proper principal submatrix of both $S^+$ and $S^-$ is triangular
with negative diagonal; thus Hurwitz-stable. This concludes that both
matrices are unstable cores: $S^+$ an unstable-positive feedback, $S^-$ an
unstable-negative feedback.

\begin{lemma}
  \label{lem:feedback}
  Unstable-positive feedbacks have a single real-positive
  eigenvalue, unstable-negative feedbacks have no real-positive
  eigenvalues.
\end{lemma}
\begin{proof}
  The characteristic polynomial of a real $|M|\times|M|$ matrix $A$ can be
  written as $p(\lambda)= \sum_{k=0}^{|M|} (-1)^k c_k \lambda^{|M|-k}$
  where $c_k$ is the sum of the principal minors of size $k$. If $A$ is an
  unstable core, then no proper principal submatrix of $A$ has one single
  positive eigenvalue, and thus the sign of every principal minor of size
  $k<|M|$ is either $(-1)^k$ or $0$. Thus $(-1)^k c_k\ge 0$. The
  coefficients of the characteristic polynomial $p$ therefore exhibit at
  most one sign change.  All coefficients have the same sign if
  $(-1)^{|M|}\det A>0$ and thus $\sign\det A = (-1)^{|M|}$, i.e., if $A$ is
  an unstable-negative feedback.  Otherwise, $(-1)^{|M|}\det A<0$ and thus
  $\sign\det A = (-1)^{|M|-1}$ and $A$ is an unstable-positive
  feedback. Descartes' Rule of Sign, see e.g.\ \cite{Wang:04}, states that
  the number of positive roots of a polynomial is either the number of
  sign-changes between consecutive coefficients, ignoring vanishing
  coefficients, or is less than it by an even number. Thus an unstable core
  has a single real-positive eigenvalue if it is an unstable-positive
  feedback and no real-positive eigenvalues if it is an unstable-negative
  feedback.
\end{proof}

We emphasize that Lemma~\ref{lem:feedback} does not exclude additional
unstable pairs of complex conjugated eigenvalues. In particular, the
unstable dimension of an unstable-positive feedback may be greater than
one. For an example, see SM Section 4. It is necessarily odd, comprising one real eigenvalue and pairs of
complex conjugated eigenvalues.
Consider {finally} the matrix property $\mathbb{P}_{\Re}$ of having one real-positive
eigenvalue. We refer to Child-Selection matrices $S[\pmb{\kappa}]$ that are
minimal w.r.t.\ $\mathbb{P}_{\Re}$ as \emph{generalized-unstable-positive
feedbacks}. {We include in the SM, Section 2, few further observations on generalized-unstable-positive feedbacks}.

\section{Autocatalytic cores as unstable cores and Metzler matrices}\label{sec:autocatMetzler}

A reaction $j_A$ is \emph{explicitly autocatalytic} if it is explicitly
catalytic and $0<{s}^{j_{A}}_m < \tilde{s}^{j_{A}}_m$, for a species
$m$. Autocatalysis can also be distributed over a sequence of reactions
that collectively exhibit ``network autocatalysis''
\cite{Smith:09,Smith:13}.  The literature does not provide a single, widely
accepted definition of ``network autocatalysis'', see \cite{DefAut20}.
Here we are inspired by the approach by Blokhuis, Lacoste, and Nghe
\cite{blokhuis20}, albeit with some adjustments in formalism and
terminology.

\begin{definition}\label{def:autocatalysis}
  Let $\kappa \subseteq M$, $\iota \subseteq E$. A $|\kappa| \times
  |\iota|$ submatrix $S'$ of the stoichiometric matrix $S$ is an
  \emph{autocatalytic matrix} if the following two conditions are
  satisfied:
  \begin{enumerate}
  \item there exists a positive vector $v\in\mathbb{R}^{|\iota|}_{>0}$ such
    that $S'v>0$,
  \item for every reaction column $j$ there exist entries $m,\tilde{m}$, not
    necessarily distinct, such that $m$ is a reactant of $j$, i.e.,
    $s^j_m>0$ and $\tilde{m}$ is a product of $j$,
    i.e. $\tilde{s}^j_{\tilde{m}}>0$.
  \end{enumerate}
\end{definition}

For the case $m=\tilde{m}$, Def.~\ref{def:autocatalysis} recovers the
definition of an explicitly-autocatalytic reaction $j_A$, where $S'$ is a
$1\times 1$ matrix $S'=S_{mj_A}=\tilde{s}^{j_A}_m-s^{j_A}_m>0$. Analogously
to Def.~\ref{def:unstablecore} we can consider `core' matrices that are
minimal with the property of being autocatalytic.
\begin{definition}
  A stoichiometric submatrix $A$ is an \emph{autocatalytic core} if it is
  autocatalytic and it does not contain a proper submatrix that is
  autocatalytic.
  \label{def:autocore}
\end{definition}

We recall that a square matrix with non-negative off-diagonal elements is
known as a Metzler matrix (see e.g.\ \cite{Duan:21}). We are now ready to
state the main theorem.
\begin{theorem}\label{thm:autocoremain}
  A submatrix $\tilde{A}$ of the stoichiometric matrix $S$ is an
  autocatalytic core if and only if $\tilde{A}$ is a $k\times k$ matrix
  with a reordering $A$ of its columns such that $A=S[\pmb{\kappa}]$ is
  a Metzler matrix and an unstable-positive feedback.
\end{theorem}
We postpone the lengthy proof to Section~\ref{sec:NC_Proof}. 
Thm.~\ref{thm:autocoremain} directly implies the following observation and
corollary.
\begin{obs}
  A reaction $j_{aut}$, explicitly autocatalytic in a species $m$,
  is thus a $1\times 1$ autocatalytic core $A:=\tilde{s}^{j_{aut}}_m-s^{j_{aut}}_m>0$. Also
  the converse is true: If $S[\pmb{\kappa}]$ has a positive diagonal entry
  $S[\pmb{\kappa}]_{mm}>0$, then $J(m)$ is an explicitly-autocatalytic
  reaction in $m$. Minimality then implies that a $k\times k$
  autocatalytic core $A$ with $k>1$ always identifies a Metzler matrix with
  non-positive diagonal.
\end{obs}

\begin{corollary}\label{cor:noautocat}
  No unstable-negative feedback is an autocatalytic matrix.
\end{corollary}
\begin{proof}
  An unstable-negative feedback is never an unstable-positive feedback. Thus, they are never autocatalytic.
\end{proof}
We call a network $\pmb{\Gamma}$ \emph{autocatalytic} if its stoichiometric
matrix $S$ contains a submatrix that is an autocatalytic core. This can
also be expressed as a corollary from Thm.~\ref{thm:autocoremain}.
\begin{corollary}\label{cor:autocatnetwork}
  A network is autocatalytic if and only if there exists a $k$-CS
  $\pmb{\kappa}$ such that the associated Child-Selection matrix
  $S[\pmb{\kappa}]$ is an unstable Metzler matrix.
\end{corollary}
\begin{proof}
If a network is autocatalytic, Thm.~\ref{thm:autocoremain} implies the existence of a matrix $S[\pmb{\kappa}]$ that is both a Metzler
matrix and unstable. Conversely, reasoning as in the proof of
Thm.~\ref{thm:autocoremain}, if $S[\pmb{\kappa}]$ is an unstable Metzler
matrix, it contains an irreducible Metzler matrix that is also unstable,
i.e. autocatalytic, and thus $S$ contains an autocatalytic core and the
network is autocatalytic.
\end{proof}
We now draw the final dynamical conclusion from Cor.~\ref{cor:autocatnetwork} 
together with Cor.~\ref{cor:CSinst}:
\begin{corollary} 
  Every autocatalytic network admits instability.
\end{corollary}

Clearly, the converse is not true: Autocatalysis is not necessary for
instability. Both the presence of an unstable-positive and an
unstable-negative feedback are sufficient conditions for the network to
admit instability, again as a consequence of Cor.~\ref{cor:CSinst}. By
Cor.~\ref{cor:noautocat}, furthermore, unstable-negative feedbacks are
never autocatalytic. Moreover, it is easy to construct unstable-positive
feedbacks that are not Metzler matrices and thus not
autocatalytic. Consider for instance the following three examples:
\begin{equation}
S_1=\begin{pmatrix}
-1 & -2\\
-1 & -1
\end{pmatrix};\quad\quad
S_2=\begin{pmatrix}
-1&0&1\\
-1&-1&0\\
1&-1&-1
\end{pmatrix};\quad\quad
S_3=\begin{pmatrix}
-1&1&1&0\\ 
-1&-1&0&1\\ 
1&1&-1&0\\ 
1&0&1&-1
\end{pmatrix}.
\end{equation}
All three matrices are unstable with a real-positive eigenvalue, since
$\det{S_1}=-1$, $\det{S_2}=1$, and $\det{S_3}=-2$.  They do not contain
proper unstable submatrices. This is straightforward to see for $S_1$ and
$S_2$, since all the proper principal submatrices are weakly diagonally
dominant with negative diagonal, and thus by Gershgorin's circle theorem
\cite{Gers:31} no proper principal submatrix can have an eigenvalue with
a positive-real part. The same argument applies to the principal submatrices
of $S_3$ with sizes 1 and 2.  A direct computation (omitted here for
brevity) shows that none of the four principal submatrices of $S_3$ with
size 3 is Hurwitz-unstable. Hence $S_3$ does not contain an unstable proper
submatrix. Then, considering $S_1=S[\pmb{\kappa}_1]$,
$S_2=S[\pmb{\kappa}_2]$, $S_3=S[\pmb{\kappa}_3]$, for $k$-Child-Selections
$\pmb{\kappa}_1$, $\pmb{\kappa}_2$, $\pmb{\kappa}_3$, we conclude that
$S_1$, $S_2$, and $S_3$ are unstable-positive feedbacks. On the other hand,
they are not autocatalytic because they are not Metzler matrices.

The matrices $S_1$ and $S_2$ have ``twin'' autocatalytic cores 
\begin{equation}
A_1=\begin{pmatrix}
-1 & 2\\
1 & -1
\end{pmatrix};\quad\quad
A_2=\begin{pmatrix}
-1&0&1\\
1&-1&0\\
1&1&-1
\end{pmatrix},
\end{equation}
in the following sense:
\begin{definition} 
  A pair $(A,B)$ of $k\times k$ matrices is a called a \emph{twin-pair} if
  $\displaystyle 
  \prod_{m\in\kappa'} A_{m\gamma(m)}=\prod_{m\in\kappa'} B_{m\gamma(m)}$
  holds for every cyclic permutation $\gamma$ on a set $\kappa'$ of size
  $k'=1,\dots,k$.
  \label{eq:twins}
\end{definition}
Twin matrices are in particular similar as they share the same characteristic polynomial and thus the same eigenvalues, but even more is true. 
\begin{lemma}\label{lem:twin}
Let $(A,B)$ be a twin-pair of matrices and $D$ be any diagonal matrix
$D$. Then $(AD, BD)$ is a twin pair of matrices.
\end{lemma}
\begin{proof} We have  
$$\displaystyle
\prod_{m\in\kappa'} A_{m\gamma(m)}d_{mm}= \prod_{m\in\kappa'}
A_{m\gamma(m)}\prod_{m\in\kappa'} d_{mm}=\prod_{m\in\kappa'}
B_{m\gamma(m)}\prod_{m\in\kappa'} d_{mm}=\prod_{m\in\kappa'}
B_{m\gamma(m)}d_{mm}$$ whenever $\gamma$ is a cyclic permutation on $\kappa'$.
\end{proof}

The unstable-positive feedback $S_3$, on the other hand, has no twin
autocatalytic core.  Since only a single off-diagonal entry is negative,
Def.~\eqref{eq:twins} can never hold for a pair $(S_3, A_3)$ where $A_3$ is a
Metzler matrix because the parity of the negative signs can never match.
See SM, Section 3.1 and 3.2, for examples of pairs of
networks that contain twins unstable core $(S_1,A_1)$ and $(S_2,A_2)$.  The
definition of twins does not seem to offer a clear chemical interpretation
behind sharing the stability features. This underlines that similarities in
the global dynamics of networks possessing twin cores are not to be
expected, a priori.

\section{Examples}\label{sec:example}

\begin{multicols}{2}
\begin{mdframed}[backgroundcolor=yellow!15,rightline=false,leftline=false]
\textbf{EXAMPLE A: Superlinear growth with unstable-positive feedback.}\\

\begin{minipage}{0.5\columnwidth}
\begin{align*} 
       A+B &\underset{1}{\rightarrow} C\\
       2A+B &\underset{2}{\rightarrow} D\\
       C &\underset{3}{\rightarrow}\\
\end{align*}
\end{minipage}%
\begin{minipage}{0.5\columnwidth}
\begin{align*} 
       D &\underset{4}{\rightarrow}\\
       \quad\quad\quad\;\; &\underset{F_A}{\rightarrow} A\\
    \quad\quad\quad \;\;&\underset{F_B} {\rightarrow} B\\
\end{align*}
\end{minipage}\\
The unique unstable core
\begin{equation}\label{twinable1}
\begin{blockarray}{ccc}
& 1 & 2 \\
\begin{block}{c(cc)}
  A & \;-1 & -2\; \\
  B & \;-1 & -1\;\\
\end{block}
\end{blockarray}.\end{equation}
is an unstable-positive feedback. It is not a Metzler matrix, thus the
network is non-autocatalytic. The associated parameter-rich system admits
multistationarity in the form of two equilibria, one stable and one unstable.
On the heteroclinic orbit connecting the two equilibria, we see superlinear
growth of the concentration $x_A$.
 \begin{center}
\includegraphics[width=\textwidth]{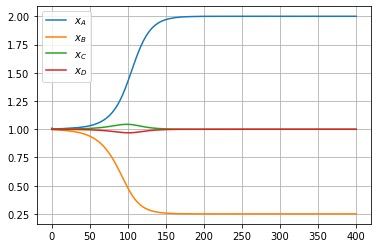}
\end{center}
See SM, Section 3.1.
\end{mdframed}

\begin{mdframed}[backgroundcolor=yellow!15,rightline=false,leftline=false]
\textbf{EXAMPLE B: Oscillations with unstable-positive feedback.}\\

\begin{minipage}{0.5\columnwidth}
\begin{align*}       
    A+B &\underset{1}{\rightarrow} C+E\\
    B+C &\underset{2}{\rightarrow} E\\
    C+D &\underset{3}{\rightarrow} A+E\\
    2B+D &\underset{4}{\rightarrow} E\\
    \end{align*}
\end{minipage}%
\begin{minipage}{0.5\columnwidth}
\begin{align*}
       E &\underset{5}{\rightarrow}\\
    &\underset{F_A}{\rightarrow} A\\
&\underset{F_B}{\rightarrow} B\\
  &\underset{F_D}{\rightarrow} D\\
\end{align*}
\end{minipage}\\
The unique unstable core
\begin{equation} \label{pfna}
\begin{blockarray}{cccc}
& 1 & 2 & 3\\
\begin{block}{c(ccc)}
A & \;-1 & 0 & 1\;\\
B & \;-1 & -1 & 0\;\\
C &  \;1 & -1 & -1\;\\  
\end{block}
\end{blockarray}  
\end{equation}
is an unstable-positive feedback. It is not a Metzler matrix and thus
the network is non-autocatalytic.
%It possesses an autocatalytic twin of the form
%\begin{equation} \label{pfna}
%\begin{pmatrix}
 %  -1 & 0 & 1\\
 %   1 & -1 & 0\\
 %   1 & 1 & -1   
%\end{pmatrix}.
%\end{equation}
The associated parameter-rich system admits oscillations.
\begin{center}
  \includegraphics[width=\textwidth]{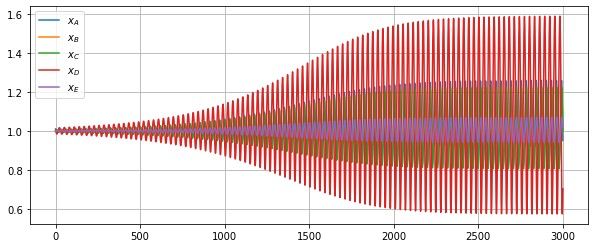}
\end{center}

\begin{center}
  \includegraphics[width=\textwidth]{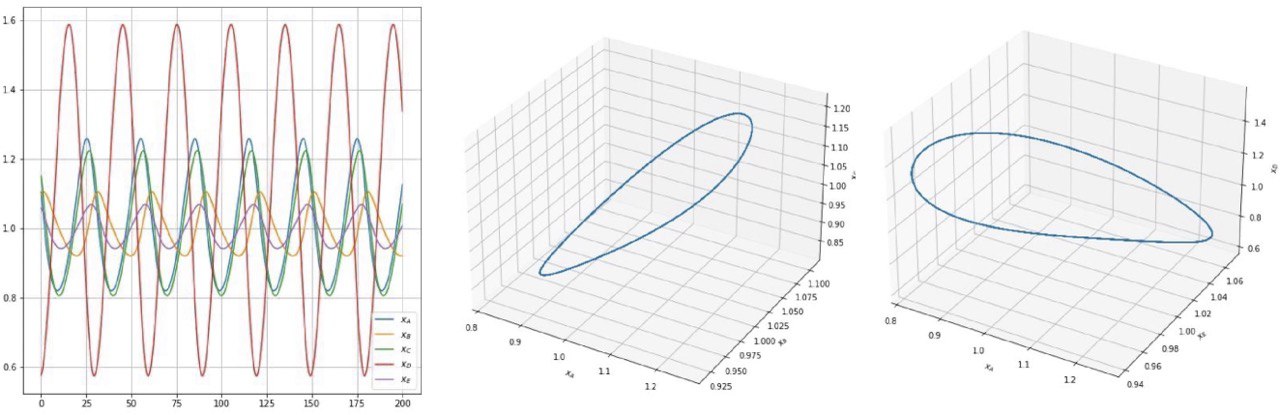}
\end{center}

See SM, Section 3.2.
\end{mdframed}
\end{multicols}

\newpage

\begin{multicols}{2}
\begin{mdframed}[backgroundcolor=yellow!15,rightline=false,leftline=false]
\textbf{EXAMPLE C: Oscillations with unstable-negative feedback.}\\
\begin{minipage}{0.5\columnwidth}
\begin{align*}       
    A+B &\underset{1}{\rightarrow} A+F\\
    B+C &\underset{2}{\rightarrow} B+F\\
    C+D &\underset{3}{\rightarrow} C+F\\
    D+E &\underset{4}{\rightarrow} D+F\\
    E+A &\underset{5}{\rightarrow} E+F\\
\end{align*}
\end{minipage}%
\begin{minipage}{.5\columnwidth}
\begin{align*}
   &\underset{F_A}{\rightarrow} A\\
&\underset{F_B}{\rightarrow} B\\
  &\underset{F_C}{\rightarrow} C\\
  &\underset{F_D}{\rightarrow} D\\
 &\underset{F_E}{\rightarrow} E\\       
    F&\underset{6}{\rightarrow}\\
\end{align*}
\end{minipage}\\
The
unique unstable core
\begin{equation*}
\begin{blockarray}{cccccc}
& 1 & 2 & 3 & 4 & 5\\
\begin{block}{c(ccccc)}
\\
  A & 0 & 0 & 0 & 0 & -1\\
  B & -1 & 0 & 0 & 0 & 0\\
  C & 0 & -1 & 0 & 0 & 0\\
  D & 0 & 0 & -1 & 0 & 0\\
  E & 0 & 0 & 0 & -1 & 0\\
  \\
\end{block}
\end{blockarray}
\end{equation*}
is an unstable-negative feedback, hence non-autocatalytic. The associated
parameter-rich system admits oscillations.
\begin{center}
\includegraphics[width=\columnwidth]{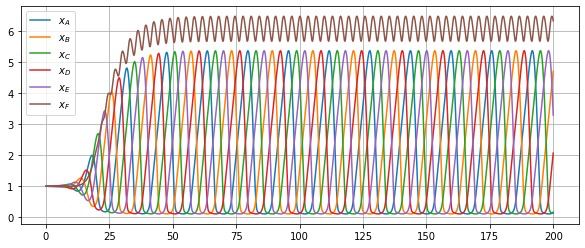}
\end{center}
\begin{center}
\includegraphics[width=\textwidth]{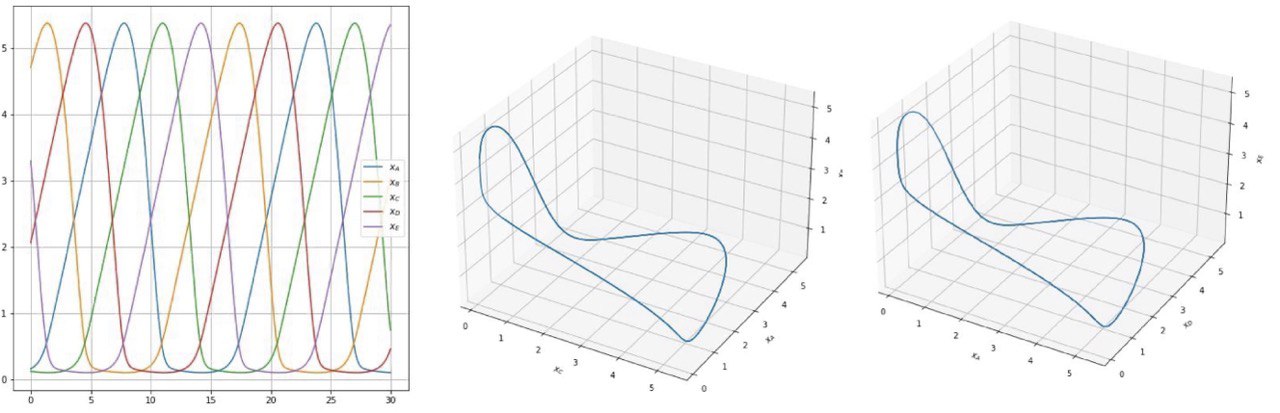}
\end{center}
See SM, Section 3.3, where we also present a mass-action variation of the same example.
\end{mdframed}

\begin{mdframed}[backgroundcolor=yellow!15,rightline=false,leftline=false]
    \textbf{EXAMPLE D: Non-autocatalytic unstable-positive feedbacks in the
      dual futile cycle.}\\
    \begin{align*} \label{doublephosphorilation}
\begin{cases}
A+E_1 \overset{1}{\underset{2}{\rightleftharpoons}} I_1 \underset{3}{\longrightarrow} B+E_1 \overset{7}{\underset{8}{\rightleftharpoons}} I_3 \underset{9}{\longrightarrow} C+E_1;\\
A+E_2 \underset{6}{\longleftarrow} I_2 \overset{5}{\underset{4}{\rightleftharpoons}} B+E_2 \underset{12}{\longleftarrow} I_4 \overset{11}{\underset{10}{\rightleftharpoons}} C+E_2 .
\end{cases}
\end{align*}
The network has only three unstable-positive feedbacks.
\begin{equation*}
\begin{blockarray}{cccccc}
& 1 & 7 & 9 & 4 & 6\\
\begin{block}{c(ccccc)}
\\
A & -1 & 0 & 0 & 0 & 1\\
E_1 & -1 & -1 & 1 & 0& 0\\
I_3 & 0 & 1 & -1 & 0 & 0\\
B & 0 & -1 & 0 & -1 & 0\\
I_2 & 0 & 0 & 0 & 1 & -1\\
\\
\end{block}
\end{blockarray}
\end{equation*}
\begin{equation*}
\begin{blockarray}{cccccc}
& 10 & 4 & 6 & 7 & 9\\
\begin{block}{c(ccccc)}
\\
C & -1 & 0 & 0 & 0 & 1\\ 
E_2 & -1 & -1 & 1 & 0& 0\\
I_2 & 0 & 1 & -1 & 0 & 0\\
B & 0 & -1 & 0 & -1 & 0\\
I_3 & 0 & 0 & 0 & 1 & -1\\
\\
\end{block}
\end{blockarray}
\end{equation*}
\begin{equation*}
\begin{blockarray}{ccccccc}
& 1 & 7 & 9 & 10 & 4 & 6\\
\begin{block}{c(cccccc)}
\\
A & -1 & 0 & 0 & 0 & 0 & 1\\
E_1 & -1 & -1 & 1 & 0& 0 & 0\\
I_3 & 0 & 1 & -1 & 0 & 0 & 0\\
C & 0 & 0 & 1 & -1 & 0 & 0\\
E_2 & 0 & 0 & 0 & -1 & -1 & 1\\
I_2& 0 & 0 & 0 & 0 & 1 & -1\\
\\
\end{block}
\end{blockarray}
\end{equation*}
They are all non-autocatalytic since they are not Metzer-matrices, but they
all admit an autocatalytic twin. Even when
endowed with a kinetics that is not parametric-rich, as Mass Action kinetics, the
network is indeed known to admit instability in the form of two stable and one
unstable equilibria \cite{HellRendall:2015}. See SM, Section 3.4.
\end{mdframed}
\end{multicols}

\section{NC-networks and what is special about autocatalysis}\label{sec:whatisspecial}

Lemma~\ref{lem:twin} together with Cor.~\ref{cor:CSinst} suggests that an
autocatalytic unstable-core is indistinguishable from a non-autocatalytic
unstable-core, as far as its (in)stability properties are
concerned. The special role of autocatalysis in chemical reaction
  networks thus is not captured by spectral properties. In this section we
  provide an intuitive explanation. To this end,
%
%  What is special about autocatalysis thus is not captured by
%spectral properties,
%  and this section provides an intuitive answer to this relevant question.
{we limit our attention to networks without explicitly
  catalytic reactions, i.e., without reactions where a species is
  both a reactant and a product. We refer to such networks as
  non-explicitly-catalytic networks (in brief, NC-networks) and we list few
  advantages of restricting attention to NC-networks. First, by Eq.~\eqref{S}, the stoichiometric matrix $S$ unambiguously represents the stoichiometry of the system.  Second, by Def.~\ref{def:kineticmodel}, the reactivity matrix $R$ has positive entries $R_{jm}>0$ if and only if $s^j_m$ is negative and thus if and only if $S_{mj}$ is negative. As a direct consequence, by Eq.~\eqref{def:csmatrixentries}, Child-Selection matrices $S[\pmb{\kappa}]$ have always strictly negative diagonal in NC-networks. Consider now a CB-summand $G([\kappa,E_\kappa])=S[\kappa,E_\kappa]R[\kappa,E_\kappa]$. Our second observation above implies that there exists a reordering of the columns of $S[\kappa,E_\kappa]$ that is a Metzler-matrix if and only if there exists a reordering of the columns of $R[\kappa,E_\kappa]$ that is diagonal. Expanding on the above considerations,} the following two corollaries of
Thm. ~\ref{thm:autocoremain} point at the peculiarity of autocatalysis among
unstable cores.
\begin{corollary}\label{cor:diagonal}
  Let $\pmb{\Gamma}^{NC}$ be a NC-network. Consider a CB-summand
$G([\kappa,E_\kappa])=S[\kappa,E_\kappa]R[\kappa,E_\kappa]$.  If
  $S[\kappa,E_\kappa]$ is an autocatalytic core, then
\begin{equation}\label{eq:cbautinst}
    \det R[\kappa,E_\kappa]=\sgn(J) \prod_{m\in\kappa} r'_{J(m)m} \quad
    \text{and} \quad \sign(\det S[\kappa,E_\kappa])=\sgn(J)(-1)^{k-1},
  \end{equation} 
  for a unique Child-Selection bijection $J$.
\end{corollary}
\begin{proof}
{
  By Thm.\ref{thm:autocoremain}, there exists a reordering $A$ of the columns of $S[\kappa,E_\kappa]$ such that $A$ is a Metzler matrix and an unstable-positive feedback. 
  Since $\pmb{\Gamma}^{NC}$  is a NC-network, there exists a diagonal
  reordering of $R[\kappa,E_\kappa]$. The reordering follows precisely from
  the unique Child-Selection bijection $J$ that associates a species $m$ in $\kappa$
  to the unique reaction $J(m)$ in $E_\kappa$ with $s^{J(m)}_m<0$. In particular, there is a unique CS-triple $\pmb{\kappa}=(\kappa,E_\kappa, J)$ defined on the pair of sets ($\kappa$, $E_\kappa$). Thus,
  since the determinant is an alternating form, we conclude that $\det R[\kappa,E_\kappa]=\sgn(J) \prod_{m\in\kappa} r'_{J(m)m}$. Since the
  reordering $A=S[\pmb{\kappa}]$ of $S[\kappa,E_\kappa]$ is an
  unstable-positive feedback,  \eqref{eq:jperm} implies that $\sign(\det S[\kappa,E_\kappa])=\sgn(J)(-1)^{k-1}$.}
\end{proof}
Corollary \ref{cor:diagonal} implies that the stability of the CB-summand
associated with the stoichiometry of an autocatalytic core corresponds to the
stability of a single elementary CB-component. An interesting consequence
of Cor.~\ref{cor:diagonal} is the following.

\begin{corollary}\label{cor:alwaysunstable}
  Let $\pmb{\Gamma}^{NC}$ be a NC-network. Consider a CB-summand
  $G([\kappa,E_\kappa])=S[\kappa,E_\kappa]R[\kappa,E_\kappa]$, where
  $S[\kappa,E_\kappa]$ is an autocatalytic core. Then,
  $G([\kappa,E_\kappa])$ is unstable for any choice of symbols $r'_{jm}$
  with $m\in \kappa$, $j\in E_\kappa$.
\end{corollary}
\begin{proof}
  From Cor.~\ref{cor:diagonal} we compute the sign of the determinant of
  the CB-summand $G([\kappa,E_\kappa])$,
  \begin{equation}
    \sign \det G([\kappa,E_\kappa])=\sign( \det S[\kappa,E_\kappa] \det
    R[\kappa,E_\kappa]) = \sgn(J) \sgn(J) (-1)^{k-1} = (-1)^{k-1},
\end{equation}
  and thus $G([\kappa,E_\kappa])$ is always unstable, independently of the
  choice of symbols $r'_{jm}$.
\end{proof}

In essence, autocatalysis provides an ``always-unstable'' subnetwork,
independent of any parameter choice. This has relevant consequences for the
dynamical analysis, in particular when dealing with realistic parameter
values or with kinetic models that are not parameter-rich, such as mass
action. {For non-autocatalytic unstable-cores, in fact, the
  subnetwork identified by the associated CB-summand does not need to
    be unstable for all parameter values: It requires a careful tuning of
  the parameters to achieve instability, as addressed in
  Cor.~\ref{cor:CSinst}, and such tuning may not be possible in contexts
  that are not parameter-rich.} We further remark, however, that the
converse of Cor.~\ref{cor:alwaysunstable} is not true: There are also
non-autocatalytic subnetworks whose associated Jacobian is unstable for
every choice of parameters. To see this, consider the stoichiometric matrix
$S_3$, already discussed in the previous section, and associated reactivity
matrix $R_3$
\begin{equation}S_3=\begin{pmatrix}
-1&1&1&0\\ 
-1&-1&0&1\\ 
1&1&-1&0\\ 
1&0&1&-1
\end{pmatrix}
  \quad\quad\quad
  R_3=
  \begin{pmatrix}
    r'_{1A} & r'_{1B} & 0 & 0\\
    0 & r'_{2B} & 0 & 0\\
    0 & 0 & r'_{3C} & 0\\
    0 & 0 & 0 & r'_{4D}
  \end{pmatrix}.
\end{equation}
Since $\det S_3=-1$ and $\det R_3=r'_{1A}r'_{2B}r'_{3C}r'_{4D}>0$, we have that $\sign \det G_3=\sign (\det S_3 \det R_3) = -1$, and thus $G_3$
is unstable irrespective of the choices of symbols.

As a final caveat, we caution the reader about exchanging statements
between {NC-networks and general networks.
Such exchanges may work well for considerations based on the stoichiometry,
as we will address in Section \ref{sec:NC_Proof}. In general, however, the
exchanges fail for the dynamics, which also depends explicitly on the reactivity matrix $R$}. In particular,
both Cor.~\ref{cor:diagonal} and Cor.~\ref{cor:alwaysunstable} do not hold
for general networks with explicitly-catalytic reactions. {The only viable way to extend this point of view to general networks would be to enforce the statement of Cor.~\ref{cor:diagonal} as the starting definition of autocatalytic cores, and we do not pursue this approach here.
Still, even in the present context, both corollaries} offer
a qualitative hint and explanation of the special nature of autocatalysis
and a possible reason for the omnipresence of autocatalysis in networks
that have been of interest in the literature.

\section{Conclusion}\label{sec:conclusion}
For a broad class of kinetic models, which includes Michaelis--Menten and
Generalized Mass Action kinetics (but not classical Mass Action kinetics),
we have shown here that {an inspection of the topology of a chemical
  reaction network may be conclusive on whether the network admits
  dynamical instability}. More precisely, we found that unstable cores,
characterized as certain minimal submatrices of the stoichiometric matrix
that are Hurwitz unstable, provide a sufficient condition for network
instability by Cor.~\ref{CSinstability}.  Moreover, we conjecture that the
  slightly more general D-unstable cores are even necessary; Conjecture
  \ref{conjecture}.

The present study thus complements investigations into sufficient
conditions for universal stability as \emph{deficiency zero} \cite{Fei87},
or exclusion of multiple equilibria as \emph{injectivity} \cite{Ba-07,
  BaCra10} and \emph{concordance} \cite{ShiFei12}. See \cite{Fei19, BaPa16}
for two comprehensive overviews on such topics. Note, however, that
injectivity and concordance only concern the uniqueness of equilibria, but
not their stability, while the conditions for the deficiency-zero theorem
are quite restrictive and they only hold for weakly-reversible mass-action
networks. See \cite{Muller:12} for generalizations and failures of the
deficiency statements for a parameter-rich model such as Generalized Mass
Action kinetics.

We observed that there are two classes of unstable cores, unstable-positive
and unstable-negative feedbacks, distinguished by the sign of their
determinant as well as the presence of real-valued positive eigenvalue.
Autocatalytic cores, as defined in \cite{blokhuis20}, turn out to be
exactly (after a suitable permutation
of the columns) unstable-positive feedbacks that are Metzler matrices. 
This simple characterization lends support to this particular definition of autocatalysis, which also has been
adapted e.g.\ in \cite{Despons:23,Gagrani:23}.

While positive and negative feedbacks are distinguished by spectral
properties, this is not the case for autocatalytic versus non-autocatalytic
unstable cores. The existence of twins shows that there are both types of
cores sharing the same characteristic polynomial.  Nevertheless,
autocatalysis is the source of instability in well-known examples of
chemical networks with ``interesting'' dynamics. This is probably a
consequence of the fact that an autocatalytic core is dynamically unstable
irrespective of the parameter choices (Cor.~\ref{cor:diagonal} and
Cor.~\ref{cor:alwaysunstable}). In order to design an example of dynamic
instability, it therefore suffices to pick an autocatalytic core and
complement it with feed and waste product to achieve a modicum of chemical
realism. This, together with chemists' intuition on the importance of
autocatalysis, may explain why non-autocatalytic examples do not appear to
be widely known.

\section{Proof of Theorem \ref{thm:autocoremain}}
\label{sec:NC_Proof}
{For a proof of Thm.~\ref{thm:autocoremain},} it is
convenient to remove momentarily explicit catalysis from the network.  {As addressed in Sec.~\ref{sec:whatisspecial}}, this
has the advantage that the definition of autocatalysis then relies on the
sign structure of the matrix alone and does not require a reference to the
numerical values of the stoichiometric coefficients. To this end, we split
every explicitly-catalytic reaction $j_C$, with $s^{j_C}_m\tilde{s}^{j_C}_m
\neq 0$ for at least one species $m$, into two reactions $j_{C1}$, $j_{C2}$
such that $s^{j_{C1}}_m=s^{j_C}_m$, and
$\tilde{s}^{j_{C2}}_m=\tilde{s}^{j_C}_m$, for every $m$. That is, the
reactants of $j_C$ coincide with the reactants of $j_{1C}$ and the products
of $j_C$ coincide with the products of $j_{2C}$. An \emph{intermediate
species} $m_I$ is added, so that $m_I$ is the single product of reaction
$j_{1C}$ and the single reactant of reaction $j_{2C}$. For instance, the 
explicitly-catalytic reaction $j_C$
\begin{equation}
   m_1 + m_2 \overset{j_C}{\longrightarrow} m_1 + m_3
\end{equation}
is split into $j_{C1}$, $j_{C2}$ as
\begin{equation}
  m_1 + m_2 \overset{j_{C1}}{\longrightarrow} m_I
  \overset{j_{C2}}{\longrightarrow}m_1 + m_3
\end{equation}
with the addition of the intermediate species $m_I$. We call the
replacement of an explicitly-catalytic reaction $j_C$ with the
\emph{non-explicitly-catalytic triple} $(m_I,j_{C1},j_{C2})$ a
\emph{non-explicitly-catalytic extension} of $j_C$. Clearly, we can
generalize this procedure to all explicitly-catalytic reactions in the
network.
\begin{definition}[Non-explicitly-catalytic extension]
  A network $\pmb{\Gamma}^{NC}$ is a \emph{non-explicitly-catalytic
  extension} (in short, NC-extension) of $\pmb{\Gamma}$ if
  $\pmb{\Gamma}^{NC}$ is obtained from $\pmb{\Gamma}$ by substituting all
  explicitly-catalytic reactions $j_C$ with triples $(m_I,j_{C1},j_{C2})$.
  Let $S$ be the stoichiometric matrix of $\pmb{\Gamma}$. Then $S^{NC}$
  indicates the stoichiometric matrix of $\pmb{\Gamma}^{NC}$.
\end{definition}

\begin{obs}
  Let $\pmb{\Gamma}$ be a network with $|M|$ species and $|E|$ reactions,
  of which $n$ are explicitly-catalytic. Then $\pmb{\Gamma}^{NC}$ is a
  network with $(|M|+n)$ species and $(|E|+n)$ reactions. Moreover, if
  $S$ is the $|M|\times|E|$ stoichiometric matrix of $\pmb{\Gamma}$, the
  stoichiometric matrix $S^{NC}$ of $\pmb{\Gamma}^{NC}$ is an
  $(|M|+n)\times(|E|+n)$ matrix. In particular, if $S$ is a square matrix,
  so it is $S^{NC}$.
\end{obs}
The notion of NC-extension is consistent with the definition of
autocatalytic core. Let $A$ be an autocatalytic core in $\pmb{\Gamma}$ and
call $A^{NC}$ the NC-extension of $A$ in $\pmb{\Gamma}^{NC}$. We have the
following lemma.
\begin{lemma}\label{lem:extension}
  $A$ is an autocatalytic core of $\pmb{\Gamma}$ if and only if $A^{NC}$ is
  an autocatalytic core of $\pmb{\Gamma}^{NC}$.
\end{lemma}
\begin{proof}
  Let $A$ be a $|\kappa| \times |\iota|$ stoichiometric submatrix of
  $\pmb{\Gamma}$. Without loss of generality, we can consider that
  $\pmb{\Gamma}$ has one single explicitly-catalytic reaction $j_C$ in
  $\iota$, since we can inductively iterate the argument for any
  explicitly-catalytic reaction. Then, let $(m_I, j_{1C}, j_{2C})$ be its
  associated non-explicitly-catalytic triple. $A^{NC}$ is a
  $(|\kappa|+1)\times(|\iota|+1)$ matrix.

  Assume first that $A$ is an autocatalytic matrix of
  $\pmb{\Gamma}$. Consider the positive vector $v\in
  \mathbb{R}^{|\iota|}_{>0}$ such that $Av>0$ and define
  $\tilde{v}(\varepsilon)\in \mathbb{R}^{|\iota|+1}_{>0}$ as
  \begin{equation}
  \tilde{v}_j(\varepsilon):=
  \begin{cases}
    v_j \quad\quad\quad\;\; \text{for $j\neq j_{1C}$, $j\neq j_{2C}$}\\
    v_{j_C}+\varepsilon\quad\;\text{for $j=j_{1C}$} \\
    v_{j_C} \quad\quad\quad \text{for $j=j_{2C}$}
  \end{cases},
  \end{equation}
  for positive $\varepsilon>0$. This implies that
  $(A^{NC}\tilde{v})_m=(Av)_m - \varepsilon$ for $m\neq m_I$ with
  $s^{j_C}_m=s^{j_{1C}}_m>0$, $(A^{NC}\tilde{v})_m=(Av)_m$ for $m\neq m_I$
  with $s^{j_C}_m=s^{j_{1C}}_m=0$, and
  $(A^{NC}\tilde{v})_{m_I}=\varepsilon$. In particular, for small enough
  $\varepsilon>0$, we have that $A^{NC}\tilde{v}>0$. Moreover, the column
  $j_C$ of $A$ has an entry $m$ with $s^{j_C}_m>0$. This implies that the
  entry $m$ of the column $j_{1C}$ of $A^{NC}$ satisfies
  $s^{j_{1C}}_m>0$. On the other hand the entry $m_I$ of the column
  $j_{1C}$ satisfies $\tilde{s}^{j_{1C}}_{m_I}>0$. Similarly, the column
  $j_{C}$ has an entry $\tilde{m}$ with $\tilde{s}^{j_C}_{\tilde{m}}>0$;
  thus the entry $\tilde{m}$ of the column $j_{2C}$ of $A^{NC}$ satisfies
  $\tilde{s}^{j_{2C}}_{\tilde{m}}>0$. On the other hand the entry $m_I$ of
  the column $j_{2C}$ satisfies $s^{j_{2C}}_{m_I}>0$, concluding that
  $A^{NC}$ is an autocatalytic matrix of $\pmb{\Gamma}^{NC}$.

  Assume now that $A^{NC}$ is an autocatalytic matrix of
  $\pmb{\Gamma}^{NC}$. There exists a positive vector $\tilde{u} \in
  \mathbb{R}^{|\iota|+1}_{>0}$ with $A^{NC}\tilde{u}>0$. Consider the
  positive vector $\tilde{v}\in \mathbb{R}^{|\iota|+1}_{>0}$ defined as
  $\tilde{v}_j:=\tilde{u}_j$ for all $j\neq j_{1C}$ and
  $\tilde{v}_{j_{1C}}:=\tilde{v}_{j_{2C}}=\tilde{u}_{j_{2C}}$. We show that
  $(A^{NC}\tilde{v})_m \ge (A^{NC}\tilde{u})_m>0$ for all $m\neq m_I$. In
  fact, note that $\tilde{u}_{j_{1C}}>\tilde{u}_{j_{2C}}$, otherwise
  $(A^{NC}\tilde{u})_{m_I}\le 0$ would contradict $A^{NC}\tilde{u}>0$. On
  the other hand, the column $j_{1C}$ of $A^{NC}$ has only negative
  entries, with the single exception of the entry $m_I$ with $A^{NC}_{m_I
    j_{1C}}=1$. This yields
  $\tilde{v}_{j_{2C}}(A^{NC}_{mj_{1C}}+A^{NC}_{mj_{1C}})\ge
  \tilde{u}_{j_{1C}}A^{NC}_{mj_{1C}}+\tilde{u}_{j_{2C}}A^{NC}_{mj_{2C}}$
  and thus $(A^{NC}\tilde{v})_m\ge(A^{NC}\tilde{u})_m>0$ for all $m\neq
  m_I$. Since $m_I$ is not a species in $\pmb{\Gamma}$, it does not appear
  as a row in $A$. Thus we conclude that for the positive vector $v \in
  \mathbb{R}^{|\iota|}_{>0}$ defined as $v_j:=\tilde{v}_j$ for $j\neq j_C$
  and $v_{j_C}:=\tilde{v}_{j_{1C}}=\tilde{v}_{j_{2C}}$ we obtain $A v >
  0$. Moreover, the column $j_{1C}$ of $A^{NC}$ contains an entry $m\neq
  m_I$ with $s^{j_{1C}}_m>0$ and the column $j_{2C}$ contains an entry
  $\tilde{m}\neq m_I$ with $\tilde{s}^{j_{2C}}_m>0$. This implies that the
  column $j_C$ of $A$ has as well the entry $m$ with $s^{j_C}_m>0$ and the
  entry $\tilde{m}$ with $\tilde{s}^{j_C}_{\tilde{m}}>0$, and concludes
  that $A$ is an autocatalytic matrix of $\pmb{\Gamma}$.

  The above proves that a matrix $A$ is autocatalytic in $\pmb{\Gamma}$ if
  and only if $A^{NC}$ is autocatalytic in $\pmb{\Gamma}^{NC}$. The
  preservation of minimality follows from a straightforward
  observation. $B$ is a submatrix of $A$ in $\pmb{\Gamma}$ if and only if
  $B^{NC}$ is a submatrix of $A^{NC}$ in $\pmb{\Gamma}^{NC}$. Thus, $A$
  does not contain any autocatalytic proper submatrix if and only if
  $A^{NC}$ does not contain any autocatalytic proper submatrix.
\end{proof}

Since the autocatalytic cores of a network $\pmb{\Gamma}$ can be
studied in its NC-extension $\pmb{\Gamma}^{NC}$ as a consequence of
Lemma~\ref{lem:extension}, we proceed now by considering NC-networks, only. The
advantage, {we repeat,} is that in NC-networks the stoichiometric matrix $S$ fully
determines the reactivity matrix $R$, $R_{jm} \neq 0$ if and only if
$S_{mj}<0$. A first consequence of this approach is that
autocatalytic matrices and cores can be defined purely at the matrix level,
without reference to the stoichiometric coefficients. We have indeed the
following proposition.

\begin{prop}\label{prop:autocatalysisncat}
  Consider a NC-network $\pmb{\Gamma}^{NC}$. Let $\kappa \subseteq M$,
  $\iota \subseteq E$. Then a $|\kappa| \times |\iota|$ submatrix $A$ of
  the stoichiometric matrix $S$ is an \emph{autocatalytic core} if the
  following three conditions are satisfied:
  \begin{enumerate}
  \item there exists a positive vector $v\in\mathbb{R}^{|\iota|}_{>0}$ such
    that $Av>0$,
  \item for every column $j$ there exist both positive and negative
    entries, i.e., there is $m,\tilde{m}$ such that $A_{mj}>0$ and
    $A_{\tilde{m}j}<0$,
  \item $A$ does not contain any proper submatrix satisfying conditions (i)
    and (ii).
  \end{enumerate}
\end{prop}
\begin{proof}
  For NC-networks, a reactant $m$ of a reaction $j$ always appears as a
  strictly negative entry $A_{mj}<0$, whereas a product $\tilde{m}$ of
  reaction $j$ always appears as a strictly positive entry
  $A_{\tilde{m}j}>0$. The rest follows from the definition of
  autocatalytic cores.
\end{proof}
Prop.~\ref{prop:autocatalysisncat} is clearly equivalent to
Def.~\ref{def:autocore} for the case of NC-networks. In \cite{blokhuis20}
the analysis focused on NC-networks, and the conditions in
Prop.~\ref{prop:autocatalysisncat} were indeed given as the definition of
autocatalytic cores. We further note that in NC-networks autocatalysis
requires at least two species and two reactions, since explicit
(auto)catalysis is excluded.
\begin{lemma}\label{lemma:size2}
  If $S$ is a $|\kappa|\times|\iota|$ autocatalytic matrix of a NC-network,
  then $|\kappa|\ge2$ and $|\iota|\ge 2$.
\end{lemma}
\begin{proof}
  Condition 2 in Prop.~\ref{prop:autocatalysisncat} implies that there are
  at least 2 rows: $|\kappa|\ge 2$. Then $|\iota|=1$ is impossible since we
  would have $S_{m1}<0$, and thus condition 1 could not be satisfied; hence
  $|\iota|\ge 2$.
\end{proof}

For NC-networks, \cite{blokhuis20} proves that autocatalytic cores necessarily satisfy relevant conditions, as stated in the following Lemma:
\begin{lemma}[Blokhuis et al. \cite{blokhuis20}] \label{lem:Nghe}
  Let $\tilde{A}$ be an autocatalytic core of a NC-network. Then
  $\tilde{A}$ is square, invertible, and irreducible. Moreover, there
  exists an autocatalytic core $A$, obtained by re-ordering the columns of
  $\tilde{A}$, such that $A$ is a Metzler matrix with negative diagonal.
\end{lemma}

For completeness, we prove Lemma \ref{lem:Nghe} in our own setting and
notation in the SM, Section 1.3. A straightforward
corollary follows.
\begin{corollary}\label{cor:unique}
  Let $\tilde{A}$ be an autocatalytic core of a NC-network. Then, there
  exists a \emph{unique} autocatalytic core $A$ with negative diagonal
  obtained by reordering the columns of $\tilde{A}$.
\end{corollary}
\begin{proof}
  Since $A$ is a Metzler matrix, each column and row has a unique negative
  entry. Thus each column and each row of $\tilde{A}$ has a unique negative
  entry, and the permutation of the columns with negative diagonal is
  uniquely defined.
\end{proof}

In summary, an autocatalytic core of a NC-network uniquely identifies an
invertible, irreducible Metzler matrix $A$ with negative
diagonal. Hurwitz-stability of Metzler matrices has been extensively
studied in the literature, see e.g.\ \cite{Narendra:10} and the references
therein. The notorious Frobenius--Perron theorem is typically stated for
non-negative matrices $N$, i.e., with $N_{mj} \ge 0$ for all $m$ and
$j$. By considering $\alpha \ge\max_m |A_{mm}|$ and the non-negative matrix
$N:=A+\alpha \operatorname{Id}$, some important consequences of the theorem
generalize to Metzler matrices.  In particular, the Frobenius--Perron
theorem guarantees that any Metzler matrix $A$ has one real eigenvalue
$\lambda^*$ such that $\lambda^*>\Re\lambda$ for all other eigenvalues
$\lambda$. Moreover, if $A$ is irreducible, then the eigenvector $v$ of
$\lambda^*$ can be chosen positive. Another straightforward consequence is
that $Av\ge av$ for some $a>0$ implies $\lambda\ge a$. See e.g.
\cite{Bullo20}, where also the following consequent statement can be
found. 
\begin{lemma}\label{Metzlerlemma}
  Let $A$ be an invertible, irreducible, Metzler matrix. Then $A$ is
  Hurwitz-unstable if and only if there exists a positive $v>0$ such that
  $Av>0$.
\end{lemma}
To make this document self-contained, we include a short proof in the SM, Section 1.4. Lemma \ref{Metzlerlemma} implies that
within the set of invertible and irreducible Metzler matrices with negative
diagonal, Hurwitz-instability characterizes autocatalytic
matrices. We expand on such an idea in the last lemma of this section, where
we characterize autocatalytic cores among the Metzler matrices with
negative diagonal.

\begin{lemma}\label{minimality}
  Let $A$ be a $k\times k$ Metzler matrix with
  negative diagonal. Then the following two statements are equivalent:
  \begin{enumerate}
  \item $A$ is an autocatalytic core;
  \item $\sign\det A=(-1)^{k-1}$ and
    for every $\emptyset\ne\kappa'\subsetneq\{1,\dots,k\}$ we have  
    $\sign\det A[\kappa']=(-1)^{|\kappa'|}$ or
    $\det A[\kappa']=0$. 
  \end{enumerate}
\end{lemma}
\begin{proof}
  Suppose $A$ is an autocatalytic core. The proof of Lemma
  \ref{Metzlerlemma} shows that $A$ is Hurwitz-unstable with at least one
  real-positive eigenvalue $\lambda^*$. Descartes' rule of sign therefore
  implies that the coefficients $c_n$ of the characteristic polynomial
  $p(\lambda)= \sum_{n=0}^k (-1)^n c_n \lambda^{k-n}$ of $A$ do not all
  have the same sign. Since $c_1=\tr A<0$, there exists a smallest $n$,
  $1<n\le k$ and $\kappa'\subset\{1,2,\dots,n\}$ with $|\kappa'|=n$ such
  that the submatrix $A[\kappa']$ satisfies $\sign\det
  A[\kappa']=(-1)^{|\kappa'|-1}$ and $\sign\det
  A[\kappa'']=(-1)^{|\kappa''|}$ or $\sign\det A[\kappa'']=0$ for all
  non-empty $\kappa''\subsetneq\kappa'$. We proceed by showing that
  $A[\kappa']$ is an autocatalytic core. If $A[\kappa']$ is not
  irreducible, the rows and columns can be simultaneously reordered such
  that $A[\kappa']$ has a block-triangular form. Thus there is
  $\kappa''\subsetneq\kappa'$ such that $\sign\det
  A[\kappa'']=(-1)^{|\kappa''|-1}$, contradicting minimality of $
  \kappa'$.  Thus $A[\kappa']$ is irreducible and thus every column
  contains a non-zero off-diagonal entry. Since $A$ is a Metzler matrix
  with negative diagonal, every column in particular contains both a
  positive and a negative entry. Since $\sign\det
  A[\kappa']=(-1)^{|\kappa'|-1}\ne 0$, $A[\kappa']$ is invertible. It is
  unstable, since by construction its characteristic polynomial has exactly
  one sign change.  Lemma~\ref{Metzlerlemma} now implies the existence of a
  positive vector $v\in\mathbb{R}^{|\kappa|'}_{>0}$ such that $A[\kappa'] v
  >0$. Thus $A[\kappa']$ is an autocatalytic core and cannot be a proper
  submatrix of the network autocatalytic core $A$. Hence $A[\kappa']=A$ and
  $|\kappa'|=k$.

  Conversely, assume that condition 2 holds. We first show that $A$ is
  invertible and irreducible. Invertibility trivially follows from $\det A
  \neq 0$. Indirectly assume now that $A$ is reducible; then the
  determinant of $A$ can be expressed as a product of determinants of
  principal submatrices, which would imply the existence of a proper
  principal submatrix $A[\kappa']$ with $\sign \det
  A[\kappa']=(-1)^{|\kappa'|-1}$, contradicting the assumption; thus $A$ is
  irreducible. The sign of its determinant further implies that $A$ is
  unstable and thus Lemma \ref{Metzlerlemma} implies the existence of a
  positive vector $v\in\mathbb{R}^{k}_{>0}$ such that $A v>0$.  Moreover,
  the irreducibility of $A$ again guarantees that every column of $A$ has
  one negative diagonal entry and at least one positive off-diagonal entry.
  Thus $A$ is autocatalytic. It remains to show that there are no
  autocatalytic submatrices. The only submatrices that could be
  autocatalytic cores are the principal ones since an autocatalytic core
  uniquely identifies a square matrix with one negative entry in each
  column by Cor.~\ref{cor:unique}. Moreover, we can restrict to the
  irreducible ones, by Lemma~\ref{lem:Nghe}. Again, all invertible
  and irreducible principal $|\kappa'|$-submatrices $A[\kappa']$ with
  $|\kappa'|<k$ satisfy $\sign\det A[\kappa']=(-1)^{|\kappa'|}$ by
  assumption. By Lemma~\ref{lem:feedback}, none of them has any positive
  real eigenvalue.  Following the arguments in the proof of Lemma
  \ref{Metzlerlemma}, $A[\kappa']$ is Hurwitz-stable and thus there is no
  positive vector $v\in \mathbb{R}^{|\kappa'|}_{>0}$ such that $A[\kappa']
  v>0$. Therefore, $A$ does not contain a proper submatrix that is
  autocatalytic and hence it is an autocatalytic core.
\end{proof}

We are now in the position to state the main result for NC-networks.
\begin{lemma}\label{autocatcorethm}
  Let $S$ be a stoichiometric matrix of a NC-network. Then a submatrix
  $\tilde{A}$ of $S$ is an autocatalytic core if and only if
  $\tilde{A}$ is a $k\times k$ submatrix of $S$ whose columns
  can be rearranged in a matrix $A$ that satisfies the following
  conditions:
  \begin{enumerate}
  \item $A$ is a Metzler matrix with negative diagonal;
  \item $\sign\det A=(-1)^{k-1}$;
  \item no principal submatrix $A[\kappa']$ with $|\kappa'|<k$ is
    Hurwitz-unstable.
\end{enumerate}
\end{lemma}
\begin{proof}
  By Cor.~\ref{cor:unique}, an autocatalytic core $\tilde{A}$ uniquely
  identifies a matrix $A$ that is a Metzler matrix with negative diagonal.
  Lemma~\ref{minimality} implies that among the Metzler matrices with
  negative diagonal, the autocatalytic cores are exactly those that satisfy
  $\sign\det A=(-1)^{k-1}$ and for every
  $\emptyset\ne\kappa'\subsetneq\{1,\dots,k\}$ we have $\sign\det
  A[\kappa']=(-1)^{|\kappa'|}$ or $\det A[\kappa']=0$. From Descartes' rule
  of sign an identical argument as in the proof of Lemma~\ref{minimality}
  implies that no proper principal submatrix $A[\kappa']$ with
  $|\kappa'|<k$ is Hurwitz-unstable.
  
  Conversely, assume that $A$ satisfies conditions 1-3. Condition 3 implies
  that for all principal submatrices $A[\kappa']$ we have $\sign \det
  A[\kappa']=(-1)^{|\kappa'|}$ or $\det
  A[\kappa']=0$. Lemma~\ref{minimality} concludes that $A$ is an
  autocatalytic core.
\end{proof}

%\section{Autocatalytic cores as unstable cores}
%\label{sec:autocatUnstablecores}
We focus on the relationship between autocatalytic cores and unstable
cores. We first note a straightforward corollary of
Lemma~\ref{lem:Nghe}.
\begin{corollary}
  \label{cor:omega}
  Let $A$ be a $k\times k$ autocatalytic core of an NC-network, with rows in
  $\kappa \subseteq M$ and columns in $\iota \subseteq E$. For any
  $m\in\kappa$, let $J(m)\in \iota$ indicate the unique reaction such that
  $A_{m\,J(m)}<0$. Then $\pmb{\kappa}:=(\kappa,\iota,J)$ is a $k$-CS.
\end{corollary}
\begin{proof}
  It is sufficient to notice that $A_{mJ(m)}<0$ implies $s^{J(m)}_m>0$. The
  statement follows from the definition of $k$-Child-Selection triples.
\end{proof}

Thus the notion of autocatalytic cores in \cite{blokhuis20} is consistent
with our notion of cores as minimal Child-Selections, even though the
definition of autocatalytic cores does not require a fixed order of the
columns: Any $\tilde{A}$ obtained by reordering the columns of an
autocatalytic core $A$ is also an autocatalytic core. In turn, the ordering
of the columns in the Child-Selection perspective plays a crucial role to
draw dynamical conclusions. In fact, as an immediate consequence of Cor.~\ref{cor:omega} we get that
$A:=S[\pmb{\kappa}]$ is the unique matrix with negative diagonal obtained
by reordering the columns of an autocatalytic core $\tilde{A}$. In
particular, $A$ is a Metzler matrix.

We are finally ready to prove Theorem \ref{thm:autocoremain}

\begin{proof}[Proof of Theorem \ref{thm:autocoremain}]
  Assume firstly $A=S[\pmb{\kappa}]$ is a Metzler matrix and an
  unstable-positive feedback; thus $A$ is irreducible, by
  Obs.~\ref{obs:irr}.  Lemma~\ref{Metzlerlemma} implies that there exists a
  positive vector $v\in\mathbb{R}^k_{>0}$ such that $Av>0$. By
  construction of $S[\pmb{\kappa}]$, the diagonal entries correspond to
  positive stoichiometric coefficients $s^{J(m)}_m>0$. If $k=1$, the
  instability of the $1\times1$ matrix $S[\pmb{\kappa}]$ simply implies
  that $S[\pmb{\kappa}]>0$, and thus $\tilde{s}^{J(m)}_m>0$. This concludes
  that $J(m)$ is an explicitly-autocatalytic reaction in $m$, and in
  particular $S[\pmb{\kappa}]$ is an autocatalytic core. Consider now $k\ge
  2$: $A$ is irreducible and thus every column $j$ as at least a positive
  non-diagonal entry $m$. In particular $\tilde{s}^j_m\ge A^j_m >0$; thus
  $A$ is autocatalytic. Finally, any proper principal submatrix of $A$ is
  also a Metzler matrix. By minimality of unstable-positive feedbacks, it
  does not contain any other unstable proper principal submatrix, i.e., for
  all $\kappa'$ there is no positive $v'$ such that $A[\kappa']v'>0$, again
  by Lemma~\ref{Metzlerlemma}. This implies that $A$ does not contain any
  autocatalytic submatrix, which concludes that $A$ is an autocatalytic
  core.

  Conversely assume that $A=S[\pmb{\kappa}]$ is an autocatalytic core. We
  have to show that $A$ is a Metzler matrix. Lemma~\ref{lem:extension}
  states that $A$ is an autocatalytic core if and only if its NC-extension
  $A^{NC}$ is an autocatalytic core. Using Lemma~\ref{autocatcorethm},
  $A^{NC}$ is a Metzler matrix with negative diagonal. Consider a column
  $j_C$ in $A$ associated to an explicitly-catalytic reaction $j_C$ and its
  associated non-explicitly-catalytic triple $(m_I,j_{C1},j_{C2})$ in
  $A^{NC}$. The column $j_{C1}$ has a single negative entry, i.e., a single
  reactant $m^*$ (by Lemma~\ref{lem:Nghe}) and a single product
  $m_I$, via the construction of NC-extension. Moreover, the single
  reactant $m^*$ corresponds to the diagonal entry
  $A_{m^*m^*}=S[\pmb{\kappa}]_{m^*m^*}=S_{m^*J(m^*)}$, since
  $J(m^*)=j_C$. The column $j_{C2}$ has a single negative entry, i.e., a
  single reactant $m_I$, again via the construction of NC-extension. The
  column $j_C$ in $A$, then, has entries
  $A_{mj_C}=A^{NC}_{mj_{C1}}+A^{NC}_{mj_{C2}}$, for $m \neq m_I$ since
  $m_I$ is not a species in $A$. Such sum has always non-negative summands,
  except for $m=m^*$, and $m=m_I$. The first corresponds to a diagonal
  entry in $A$, the second does not appear in $A$ at all, therefore
  $A_{mj_C}\ge 0$ for all $m\neq m^*$. Repeating such argument for all
  explicitly-catalytic reactions yields that $A_{mj}$ is non-negative for
  all non-diagonal entries $j\neq m$; thus $A$ is a Metzler
  matrix. Finally, since $A$ is autocatalytic, there exists
  $v\in\mathbb{R}^k_{>0}$ such that $Av>0$. Applying
  Lemma~\ref{Metzlerlemma} yields instability of $A$ with $\sign \det
  A=(-1)^{k-1}$ and the minimality of $A$ as an autocatalytic matrix
  translates into the minimality of $A$ as an unstable matrix, just as
  discussed above in this same proof. Hence, $A$ is an unstable-positive
  feedback and a Metzler matrix.
\end{proof}

\newpage
\begin{center}\LARGE{Supplementary Material}
\end{center}
\addcontentsline{toc}{section}{SUPPLEMENTARY MATERIAL}
\setcounter{section}{0}
The Supplementary Material (SM) is organized as follows. Section \ref{sec:SIproofsomitted} contains the proofs omitted from the main text; Section \ref{sec:SIgupf} discusses in more detail \emph{generalized-unstable-positive-feedbacks}. Section \ref{sec:SIexamples} analyzes the four examples of the main text; Section \ref{sec:SIupfex} presents an example of an unstable-positive feedback with unstable dimension greater than one.

\section*{Proofs omitted from the main text}\label{sec:SIproofsomitted}
\subsection*{Proof of Lemma~\ref{lemma:mmrich}}
\begin{proof}
Consider a positive right-kernel vector $\bar{r}\in\mathbb{R}^{|E|}_{>0}$ of the stoichiometric matrix $S$:
\begin{equation}S\bar{r}=0.\end{equation}
  Choose a sufficiently large scalar $K>0$ such that
  \begin{equation}\label{bconst}
\mathpzc{b}_m^j:=\bigg(\frac{K
      \bar{r}_j}{\bar{r}'_{jm}}\frac{s^j_m}{\bar{x}_m}-1\bigg)
    \frac{1}{\bar{x}_m}>0,
  \end{equation}
  for all $m$ with $s_{m}^{j}>0$. Set
  \begin{equation}\label{a}
    \mathpzc{a}_j \coloneqq K \bar{r}_j \prod_{m\in  M}
    \Bigg( \frac{\bar{x}_m}{(1+\mathpzc{b}^j_m \bar{x}_m)}\Bigg)^{-s^j_m},
  \end{equation}
  Inserting these parameter choices into the Michaelis--Menten nonlinearity, Eq.(\ref{MMeq}) of the main text,
$$  
  r_j(x) : =a_j\prod_{m\in M} \Bigg( \frac{x_m}{(1+b^j_m
    x_m)}\Bigg)^{s^j_m}.
$$
yields
  $r_j(\bar{x})=K\bar{r}_j>0$ and $\frac{\partial r_j(x)}{\partial
    x_m}\big|_{x=\bar{x}}=\bar{r}'_{jm}$ for all reactants $m$ of any reaction
  $j$.  In particular,
  \begin{equation*}
    Sr(\bar{x})= K \cdot S\bar{r}=0,
  \end{equation*}
  and $\bar{x}$ is an equilibrium of the
  system. Thus the conditions of Def.~\ref{def:p-rich} are satisfied.
\end{proof}
\subsection*{Detailed Cauchy-Binet analysis of Section \ref{sec:CS} of the main text}
{In this section we present the detailed Cauchy-Binet analysis  leading to Lemma \ref{lemma:CSinst} and Corollary \ref{cor:CSinst} of the main text.} The instability of a matrix is characterized by the sign of the real part
of its eigenvalues, i.e., of the roots of its characteristic polynomial. For matrices that are symbolic Jacobians of networks, the
coefficients of the characteristic polynomial can be expanded along Child
Selections: Let
\begin{equation}
    g(\lambda)=\sum_{k=0}^{|M|}(-1)^kc_k\lambda^{|M|-k}
\end{equation}
be the characteristic polynomial of the symbolic Jacobian matrix $G(
\pmb{r}')=SR(\pmb{r}')$. The coefficients $c_k$ are the sum of the
principal minors $\det{G[\kappa]}$ for all sets $|\kappa|=k$. Applying the
Cauchy--Binet formula to any principal submatrix $G[\kappa]$ of $G$ we
obtain
\begin{equation}
  \det G[\kappa] = 
  \det S[\kappa,E] R[E,\kappa]
  = \sum_{|\iota|=k}  \det
  S[\kappa,\iota] \det R[\iota,\kappa].
\end{equation}
A central observation is that all nonzero summands are associated with
Child-Selections:
\begin{lemma}\label{lemma:cb}
  Let $R[\iota,\kappa]$ be a submatrix of the reactivity matrix with 
  $|\iota|=|\kappa|=k$. Then $\det R[\iota,\kappa] \ne 0$
  only if $\iota=E_\kappa=J(\kappa)$ for some $k$-CS triple
  $(\kappa,E_{\kappa},J)$. In particular,
  \begin{equation}
    \label{eq:LeibnitzJ}
  \det R[\iota,\kappa]=\sum_{J:\kappa\mapsto
    \iota}\sgn{J}\prod_{m\in \kappa}r'_{J(m)m},
  \end{equation}
  where the sum runs on all Child-Selection bijections $J$ between $\kappa$
  and $E_\kappa:=\iota$.
\end{lemma}
\begin{proof}
  The Child-Selection bijections $J:\kappa\to\iota$ identify permutations
  of the columns; thus \eqref{eq:LeibnitzJ} coincides with the Leibniz
  formula for the determinant. To prove that this is the case, consider the
  general form of the Leibniz formula $$\det
  R[\iota,\kappa]=\sum_{\pi:\kappa\mapsto\iota}\sgn\pi
  \prod_{m\in\kappa}r'_{\pi(m)m},$$ where $\pi$ is now any bijection between
  $\kappa$ and $\iota$, not necessarily a Child-Selection bijection.  The
  statement follows by noting that $\det R[\iota,\kappa]\neq 0$ requires
  the existence of at least one permutation $\bar{\pi}:\kappa\to\iota$ with
  $\prod_{m\in\kappa}r'_{\bar{\pi}(m)m}\neq0$. Such a product is non-zero if
  and only if $r'_{\bar{\pi}(m)m}\neq 0$ for all $m\in \kappa$. By the
  definitions of $R$ and $k$-Child-Selection, this is the case if and only
  if $\bar{\pi}$ is a Child-Selection bijection $J:\kappa\to\iota$. In
  fact, $r'_{\bar{\pi}(m)m}=R_{\bar{\pi}(m)m}\neq 0$ if and only if $m$ is
  a reactant of $j=\bar{\pi}(m)$, i.e., $s^{\bar{\pi}(m)}_m\neq 0$. On the
  other hand, for a bijection $\bar{\pi}:\kappa\mapsto\iota$, the condition
  $s^{\bar{\pi}(m)}_m \neq 0$ for every $m \in \kappa$ precisely defines a
  Child-Selection bijection $J=\bar{\pi}$.
\end{proof}

Lemma~\ref{lemma:cb} implies that we can rewrite the Cauchy--Binet
expansion in the form
\begin{equation}\label{cbexpanfirst}
   \det G[\kappa]=
\sum_{E_\kappa}\det(S[\kappa,E_\kappa]R[E_{\kappa},\kappa]),
\end{equation}
where the sum runs precisely on all sets $E_\kappa \subseteq E$ for which
there exists at least one Child-Selection bijection $J$ with
$E_\kappa=J(\kappa)$. We call the matrices 
\begin{equation}\label{CBsumfirst}
G[(\kappa,E_{\kappa})]:=S[\kappa,E_\kappa]R[E_{\kappa},\kappa] 
\end{equation} 
the \emph{Cauchy--Binet (CB) summands} of $G$. We emphasize again that a
pair of sets ($\kappa$, $E_{\kappa}$) does not uniquely identify a
Child-Selection bijection, as there may be more than one bijection
$J:\kappa\to\iota$ such that $(\kappa,\iota,J)$ is a $k$-CS.  Applying
Lemma~\ref{lemma:cb}, however, we can further expand the determinant of 
CB summands along $k$-CS:
\begin{equation}\label{CBelcompfirst}
\begin{split}
  \det G[(\kappa,E_{\kappa})]&=
  \det S[\kappa,E_\kappa]
  \sum_{J:\kappa\mapsto E_\kappa}\sgn{J}\prod_{m\in \kappa}r'_{J(m)m}\\&=
  \sum_{\pmb{\kappa}} \det S[\pmb{\kappa}]
  \prod_{m\in\kappa}r'_{J(m)m}=
  \sum_{\pmb{\kappa}}\det ( S[\pmb{\kappa}]R[\pmb{\kappa}]
).
\end{split}
\end{equation}
Here, $R[\pmb{\kappa}]$ is the diagonal matrix with entries
$R[\pmb{\kappa}]_{mm}=r'_{J(m)m}.$ We call the matrices
$G[\pmb{\kappa}]:=S[\pmb{\kappa}]R[\pmb{\kappa}]$ \emph{elementary
Cauchy--Binet (CB) components} of $G$.

\begin{obs}\label{obs:label}
  Putting together \eqref{cbexpanfirst}, \eqref{CBsumfirst}, and \eqref{CBelcompfirst}, the $k^{th}$ coefficient $c_k$ of the characteristic polynomial
  $g(\lambda)$ of $G$ can be expanded along CB summands and elementary CB
  components as
\begin{equation}\label{eq:completeCBexpansionSI}
  c_k= \sum_{(\kappa,E_{\kappa})}\det G[(\kappa,E_{\kappa})] = 
  \sum_{\pmb{\kappa}}\det G[\pmb{\kappa}],
\end{equation}
where the first sum runs on all pairs of sets $(\kappa,E_\kappa)$ with
cardinality $k$ for which there exists at least one Child-Selection
bijection $J(\kappa)=E_\kappa$. The second sum runs on all $k$-CS triples
$\pmb{\kappa}$. This in particular proves that the characteristic
polynomial is independent of the labeling of the network.
\end{obs}

The following technical lemma shows that $\pmb{\Gamma}$ admits
instability for parameter-rich kinetic models whenever any CB summand
admits instability.
\begin{lemma}[Lemma \ref{lemma:CSinst} of the main text]
  Assume that $\pmb{\Gamma}=(M,E)$ is a network with a parameter-rich
  kinetic model. Assume there exists a choice of positive symbols
  $\pmb{r}'$ such that a CB summand $G[(\kappa,E_\kappa)]$ is
  Hurwitz-unstable.  Then the network admits instability.
  \label{lemma:CSinstSI}
\end{lemma}
\begin{proof}
  Since the kinetic model is parameter rich, we can choose the non-zero
  symbolic entries $\pmb{r}'$ of the reactivity matrix $R$ as a function
  of a parameter $\varepsilon$ as follows: For $m\in\kappa$ and $j\in
  E_{\kappa}$ choose $r'_{jm}>0$ independent of $\varepsilon$ and set
  $r'_{jm}= \varepsilon \rho_{jm}$ for all other $j$ and $m$ with
  $s_{m}^{j}>0$, with $\rho_{jm}>0$ any positive value. By construction,
  the square submatrix $R[E_{\kappa},\kappa]$ of $R$ comprising the rows
  $j\in E_{\kappa}$ and columns $m\in\kappa$ is independent of
  $\varepsilon$. Observation \ref{obs:label} guarantees that the stability
  does not depend on the labeling of the network. Therefore, without loss
  of generality, consider $\kappa=\{1,...,k\}$: in the limit
  $\varepsilon\to 0$, the symbolic Jacobian of $\pmb{\Gamma}$ becomes
  \begin{equation*} 
    G(\pmb{r'}(0)) = \begin{pmatrix}
      G[\kappa,E_\kappa] & 0\\ ... & 0\\
    \end{pmatrix}.
  \end{equation*}
  Thus $G(\pmb{r'}(0))$ is Hurwitz-unstable whenever the CB summand
  $G[\kappa,E_\kappa]$ is Hurwitz-unstable. Appealing to the continuity of
  eigenvalues, Hurwitz-instability of $G(\pmb{r'}(0))$ implies Hurwitz
  instability of $G(\pmb{r'}(\varepsilon))$ for sufficiently small
  $\varepsilon>0$. This in turn implies that $\pmb{\Gamma}$ admits
  instability.
\end{proof}

Moreover, fix a $k$-CS $\pmb{\kappa}=\{\kappa,E_\kappa, J\}$. We can
further choose arbitrarily small symbols $r'_{jm}=\epsilon \rho_{jm}$ with
$m\in\kappa$, $j\in E_\kappa$, and $j\ne J(m)$. The same argument as in the
proof of Lemma~\ref{lemma:CSinstSI} above yields the following.
\begin{corollary}[Corollary \ref{cor:CSinst} of the main text]
  Assume that $\pmb{\Gamma}=(M,E)$ is a network with a parameter-rich
  kinetic model. If there is a choice of positive symbols $\pmb{r}'$
  such that an elementary CB component $G[\pmb{\kappa}]$ is
  Hurwitz-unstable, then there is also a choice of positive symbols such
  that the CB summand $G[(\kappa,E_\kappa)]$ is Hurwitz-unstable. Then, in
  particular, the network admits instability.
\end{corollary}

\subsection*{Proof of Lemma~\ref{lem:Nghe}}
Here we prove Lemma \ref{lem:Nghe} of the main text. We follow in essence the arguments of Blokhuis et al. \cite{blokhuis20}, in part with alternative proofs.
\begin{lemma} [Prop.~1 in the SI of \cite{blokhuis20}]
  \label{lem:Blokhuis1}
  If $A$ is a $|\kappa| \times |\iota|$ autocatalytic core of a NC-network,
  then for every $m \in \kappa$, there are $j,j'\in \iota$ such that
  $A_{mj}<0$ and $A_{mj'}>0$.
\end{lemma}
\begin{proof}
  Since there is a positive vector $v\in \mathbb{R}^{|\iota|}_{>0}$ such
  that $\sum_jA_{mj}v_j>0$ we have $A_{mj}>0$ for some $j$. Now indirectly
  suppose there exists $m$ such that $A_{mj}\ge 0$ for all $j$. Denote by
  $P$ the set of such $m$, and consider the matrix $A'$ obtained from $A$
  by removing rows $m\in P$.  Then $\sum_j A_{m'j}v_j =\sum_j
  A'_{m'j}v_j>0$ for all $m'\not\in P$.  Moreover, for every $j$, there is
  $m'\not \in P$ such that $A_{m'j}=A'_{m'j}<0$. Denote by $I=\{j\;|
  \;A_{m'j}\le 0 \text{ for all } m'\not \in P\}$ the set of reactions for
  which all the products lie in $P$.  If $I=\emptyset$, then $A'$ is an
  autocatalytic proper submatrix of $A$, contradicting minimality of
  $A$ as an autocatalytic matrix. If $I=E$, then $\sum_{j} A_{m'j}v_j\le 0$
  for all positive vectors $v$ and all $m'\not\in P$, contradicting the
  assumption that $A$ is autocatalytic.  If $\emptyset\subsetneq
  I\subsetneq E$, denote by $A''$ the matrix obtain from $A'$ by removing
  all columns $j\in I$. Then $0<\sum_jA'_{m'j}v_j\le \sum_{j\notin I}
  A'_{m'j}v_j$, and for every $j\notin I$ there is $m'\not \in P$ such that
  $A''_{m'j}=A'_{m'j}<0$ and, by construction a $m''\not \in P$ such that
  $A''_{m''j}>0$. Hence $A''$ is an autocatalytic proper submatrix of $A$,
  again contradicting the assumption that $A$ is an autocatalytic
  core. Thus there exists no $m$ such that $A_{mj}\ge 0$ for all $j$,
  concluding the proof.
\end{proof}

\begin{lemma}
  \label{lem:removecol}
  Let $A$ be a $|\kappa|\times |\iota|$ matrix such that there is a
  positive vector $v\in\mathbb{R}^{|\iota|}_{>0}$ such that $Av>0$. If $\rk
  A< |\iota|$, then there there is a column $j$ in $A$ such that the
  $|\kappa|\times (|\iota|-1)$ matrix $A^*$ obtained by removing the column
  $j$ admits a positive vector $v^*\in\mathbb{R}^{|\iota|-1}_{>0}$ such
  that $A^*v^*>0$.
\end{lemma}
\begin{proof}
  Since $\rk A < |\iota|$, there is a non-zero right kernel vector $c$
  satisfying $Ac=0$. Thus there is $\tau\in\mathbb{R}$ such that $w=v-\tau
  c\ge 0$ and $w_j=0$ for some $j\in \iota$. Note that $j$ is not
  necessarily uniquely defined. There is, however, an open neighborhood
  $N(v)$ of $v$ in $\mathbb{R}^{|\iota|}_{>0}$ such that for all $v'\in
  N(v)$ it holds $Av'>0$ and $v'>0$. In particular, we can choose $v'\in
  N(v)$ and a unique $j\in \iota$ such that $v'_i>v_i$ for all $i\ne j$ and
  $v'_j=v_j$. The vector $v''=v'-\tau c$ now satisfies $v''_j=0$, $v''_i>0$
  for all $i\ne j$ and thus the vector $v^*\in\mathbb{R}^{|\iota|-1}$
  obtained from $v''$ by deleting row $j$ is positive, i.e., $v^*>0$.
  Moreover, we have $A^*v^* = Av'' = Av' + \tau Ac = Av' >0$.
\end{proof}
  
\begin{lemma} [Prop.~2 in the SI of \cite{blokhuis20}]
  \label{lem:Blokhuis2}
  If $A$ is a $|\kappa|\times|\iota|$ autocatalytic core of a NC-network,
  then $A$ is square and invertible. Moreover, for every species
  $m\in\kappa$ there is a unique reaction $j=j(m)\in\iota$ such that
  $A_{mj}<0$ and $A_{m j'}\ge 0$ for all $j'\ne j$.  Moreover, $j(m)$ is
  uniquely determined by $m$.
\end{lemma}
\begin{proof}
  Let $A$ be a $|\kappa|\times |\iota|$ autocatalytic core. If
  $\rk A<|\iota|$, then by Lemma~\ref{lem:removecol} above there is a column $j$
  such that the matrix $A^*$ obtained by deleting $j$ still affords a
  positive vector $v^*$ such that $(A^*v^*)_m>0$ for all $m\in
  \kappa$. Hence $A^*$ is an autocatalytic proper submatrix of $A$,
  contradicting minimality. Therefore $\rk A=|\iota|$. Now suppose there is
  $m\in \kappa$ such that for every $j\in \iota$ with $A_{mj}<0$ there is a
  species $m_j\ne m$ with $A_{m_jj}<0$ and let $A'$ be the matrix obtained
  by removing row $m$. Then $(A'v)_{m'}=(Av)_{m'}$ for all $m'\ne m$ and
  every column $j$ of $A'$ still contains a positive and a negative
  entry. Thus $A'$ is an autocatalytic proper submatrix of $A$,
  contradicting minimality. Hence, for every $m\in \kappa$ there is a
  column $j\in E$ for which $A_{mj}$ is the only negative entry. This
  implies $|\kappa|\le |\iota|$. Taken together we have $\rk A\le
  |\kappa|\le |\iota| =\rk A$, and thus $|\kappa|=|\iota|=\rk A$, i.e., $A$
  is square and invertible. Since for row $m$ there is a column $j$ such
  that $A_{mj}$ is the only negative entry in the column,
  $|\kappa|=|\iota|$ implies that $A_{mj'}\ge0$ for all $j'\ne j$ and thus
  $j=j(m)$ is uniquely determined by $m$. Moreover, every row and column of
  $A$ thus contains a unique negative entry $A_{m\,j(m)}$, and thus
  $m\mapsto j(m)$ is bijective.
\end{proof}

\begin{lemma} [Prop.~4 in the SI of \cite{blokhuis20}]
  \label{lem:Blokhuis4}
  If $\tilde{A}$ is an autocatalytic core of a NC-network, then $\tilde{A}$
  is irreducible.
\end{lemma}
\begin{proof}
  Indirectly suppose that $\tilde{A}$ is reducible.  Consider the unique
  matrix $A$ with negative diagonal obtained by reordering the columns of
  $\tilde{A}$. Clearly, $A$ is also reducible and the rows and columns can
  be simultaneously reordered such that $A=\begin{pmatrix} A' & 0 \\ X &
  Y \end{pmatrix}$, where $A'$ is an irreducible matrix corresponding to
  the subset $\kappa'\subseteq \kappa$ of reactant indices. Note that $A'$
  has negative diagonal. By assumption, $A'$ is a proper submatrix of
  $A$. Let $v$ be a positive vector such that $(Av)_m>0$ for all $m\in
  \kappa$. Then the restriction $v'$ of $v$ to $\kappa'$ is a positive
  vector such that $(A'v')_m=(Av)_m>0$ for all $m\in \kappa'$. Furthermore,
  the irreducibility of $A'$ yields that every column $j\in \kappa'$ has a
  non-zero (positive!) off-diagonal entry. Hence $A'$ is an autocatalytic
  proper submatrix of $A$, in contradiction with the assumptions. Thus $A$
  is irreducible and so it is $\tilde{A}$.
\end{proof}

Lemma~\ref{lem:Nghe} follows as a straightforward consequence of above Lemma \ref{lem:Blokhuis1}, Lemma \ref{lem:Blokhuis2}, and Lemma \ref{lem:Blokhuis4}.

\subsection*{Proof of Lemma~\ref{Metzlerlemma}}

\begin{proof}
  If $A$ is Hurwitz-unstable, then $\lambda^*>0$ and hence for its positive
  eigenvector $v>0$ we have $Av = \lambda^* v > 0$.
  For the converse, if $v>0$ is a positive vector with $Av>0$ then there
  always exists $a>0$ small enough such that the implication $Av\ge av$
  holds and hence $A$ has in particular a real-positive eigenvalue
  $\lambda\ge a>0$ and it is Hurwitz-unstable.
\end{proof}

\section*{Generalized-unstable-positive feedbacks}\label{sec:SIgupf}

{We collect here few further observations and results on generalized-unstable-positive feedbacks. First,} it is worthwhile noting that a matrix may be a generalized-unstable-positive feedback and not be an unstable core as it may contain an
unstable-negative feedback as a proper submatrix. To see this, consider the
following matrix:
\begin{equation}
    S_{g}^+=\begin{pmatrix}
    0 & 0 & -1 & 1\\
    -1 & 0 & 0 & 0\\
    0 & -1 & 0 & 0\\
    0 & 0 & 1 & 0\\
\end{pmatrix}.
\end{equation}
Clearly, $\sign\det S_{g}^+=-1=(-1)^{4-1}$ and all proper principal $k'\times k'$ submatrices $S_{g}^+[\kappa']$ satisfy either $\sign\det S_{g}^+[\kappa']=(-1)^{k'}$ or $\det S_{g}^+[\kappa']=0$. With the same reasoning of Lemma~\ref{lem:feedback} in the main text, we get that $S_{g}^+$ has one real-positive eigenvalue and no proper principal submatrix have a real-positive eigenvalue. However, the $3\times3$ principal submatrix
\begin{equation}
\begin{pmatrix}
    0 & 0 & -1\\
    -1 & 0 & 0\\
    0 & -1 & 0\\
\end{pmatrix}
\end{equation}
has eigenvalues $(-1, 0.5\pm 0.87 i)$, hence it is an unstable-negative
feedback. Using Cor.~\ref{cor:CSinst} of the main text, the presence of a
generalized-positive-unstable feedback in the stoichiometry of the network
implies that the network admits instability. Moreover, minimality
w.r.t.\ $\mathbb{P}_{\Re}$ can be checked based on the sign of the
principal minors, alone. As a consequence, finding
generalized-unstable-positive feedbacks present computational advantages
compared to finding unstable-positive feedbacks. Finally, following the
arguments in Lemma~\ref{lemma:CSinst} and Lemma~\ref{lem:feedback} in the main text, the
presence of a generalized-unstable-positive feedback characterizes the
existence of a choice of parameters such that at least one coefficient
$c_k$ of the characteristic polynomial of the Jacobian $p(\lambda)=
\sum_{k=0}^{|M|} (-1)^k c_k \lambda^{|M|-k}$ satisfies $c_k (-1)^k<0$,
which is in turn a necessary condition for multistationarity \cite{BaPa16}.\\

{As noted above, checking whether a generalized-unstable-positive feedback is an unstable-positive feedback may be computationally expensive, in practice.} 
We provide here a simple sufficient criterion that implies that a
generalized-unstable-positive feedback is an unstable-positive feedback, by
considering the related autocatalytic twin matrix, as defined in Def.~\ref{eq:twins}. {Firstly, we state a lemma, which confirms that autocatalytic generalized-positive-feedbacks are always autocatalytic core.} 
\begin{lemma}\label{cor:genaut}
  Every autocatalytic generalized-unstable-positive feedback is an
  autocatalytic core, i.e., it is an unstable-positive feedback.
\end{lemma}
\begin{proof}
  {As addressed in detail in Sec.~\ref{sec:NC_Proof} of the main text,} it is well-known that every unstable Metzler matrix has at least one real-positive eigenvalue.
  Thus, if an autocatalytic matrix $A$ is a generalized-unstable-positive
  feedback, i.e., $A$ is minimal with the property $\mathbb{P}_{\Re}$ of
  having a real-positive eigenvalue, then $A$ is minimal with the property
  of being unstable, and thus it is an unstable-positive feedback.
\end{proof}

\begin{corollary}\label{cor:twingenfeed}
Let $S$ be a generalized-unstable-positive feedback, i.e., minimal with the
property $\mathbb{P}$ of having a real-positive eigenvalue. If $S$
possesses an autocatalytic twin $A$, then $S$ is an unstable-positive
feedback.
\end{corollary}
\begin{proof}
  $A$ is a generalized-unstable-positive feedback and autocatalytic; thus it is an unstable-positive feedback via Lemma~\ref{cor:genaut}. $S$ is a
  twin-matrix of $A$, which implies that any principal submatrix of $S$
  shares the same Leibniz expansion as the respective principal submatrix
  of $A$. In particular, then, $S$ contains a proper unstable principal
  submatrix if and only if $A$ does. From the fact that $A$ is an
  unstable-core, i.e., $A$ does not contain any proper unstable principal
  submatrix, it follows that $S$ is an unstable-positive feedback.
\end{proof}

\section*{The four examples in detail}\label{sec:SIexamples}
We discuss in more detail the examples presented in \textbf{EXAMPLES A, B, C, D}
of the main text.  The examples are all non-autocatalytic and they show
multistationarity, superlinear growth, or oscillations as a consequence of
instability.  The instability can follow from two distinct features: an
unstable-positive feedback (single real-positive eigenvalue) or an
unstable-negative feedback (conjugated pairs of positive-real-part
eigenvalues).  An unstable-positive feedback can produce both
multistationarity and oscillations, while unstable-negative feedbacks do
not trigger multistationarity. For the sake of exemplification, we endow
the systems with Michaelis--Menten kinetics, which is in particular
parameter-rich. {Without a kinetics that is parameter-rich, we are not guaranteed that we can achieve any instability, of course. For example, as addressed in Section \ref{sec:parameterrich} of the main text, the spectral configurations of the equilibrium Jacobian of the associated mass action system is a subset of its spectral configurations in parameter-rich kinetic models. Thus, a dynamical feature spotted when the system is endowed with parameter-rich kinetics may not occur in mass action systems. We nevertheless provide here a mass action variation of \textbf{EXAMPLE C}, to foster discussion.} 

The ubiquitous mathematical intuition is that the presence
of an unstable core allows the detection of an unstable equilibrium for a
certain choice of Michaelis--Menten constants.  The spectral type of
instability of the Jacobian at the equilibrium, i.e. one single real
unstable eigenvalue vs a complex pair of unstable eigenvalues, suggests
multistationarity (superlinear growth) and oscillations respectively.  Note
however that the instability-type of the Jacobian need not be the same as
the instability-type of the unstable core generating it. In particular, we
stress again that an unstable-positive feedback can indeed trigger
oscillatory behavior, as it is seen in celebrated autocatalytic examples as
the \emph{Oregonator} \cite{oregonator}, or more modestly in our \textbf{EXAMPLE B} in
\ref{subsec:B}. We further warn the reader that the analysis of the
unstable core alone is inconclusive with regard to the global dynamics of
the system, for which we resort to standard numerical simulations.

\subsection*{EXAMPLE A: Superlinear growth with unstable-positive feedback}
\label{subsec:A}
Consider four species $A,B,C,D$ and six reactions:\\
\begin{minipage}{.5\textwidth}
\begin{align*}       
       A+B &\underset{1}{\rightarrow} C\\
       2A+B &\underset{2}{\rightarrow} D\\
    C &\underset{3}{\rightarrow}\\
\end{align*}
\end{minipage}%
\begin{minipage}{.2\textwidth}
\begin{align*}
       D &\underset{4}{\rightarrow}\\
       \quad\quad\quad\;\; &\underset{F_A}{\rightarrow} A\\
    \quad\quad\quad \;\;&\underset{F_B} {\rightarrow} B\\
\end{align*}
\end{minipage}\\
The stoichiometric matrix $S$ reads:
\begin{equation}
S=\begin{pmatrix}
    -1 & -2 & 0 & 0 & 1 & 0 \\
    -1 & -1 & 0 & 0 & 0 & 1\\
    1  &  0 & -1 & 0 & 0 & 0\\
    0  &  1 &  0 & -1 & 0 & 0
\end{pmatrix}.    
\end{equation}
It is straightforward to check that the matrix $S$ contains a single,
non-autocatalytic, unstable core:
\begin{equation}\label{twinable1SM}
\begin{pmatrix}
-1 & -2 \\
-1 & -1
\end{pmatrix},\end{equation}
which is an unstable-positive feedback. The idea behind the design of this
example is indeed based on the following pair of reactions:
\begin{equation}
    \begin{cases}
        A+B \underset{1}{\rightarrow}...\\ 
        2A+B \underset{2}{\rightarrow}...
    \end{cases}.
\end{equation} 
This is sufficient for the presence of the unstable-positive feedback of
\eqref{twinable1SM}.  The inflow reactions $F_A, F_B$ guarantee the
existence of a positive equilibrium. Chemical realism requires products
$C,D$ with their own degradation: outflow reactions $3, 4$. Consider the
associated dynamical system of the concentrations:
\begin{equation}\label{system_BOXA}
    \begin{cases}
    \dot{x}_A=F_A-r_1(x_A,x_B)-2r_2(x_A,x_B),\\
    \dot{x}_B=F_B-r_1(x_A,x_B)-r_2(x_A,x_B),\\  \dot{x}_C=r_1(x_A,x_B)-r_3(x_C),\\
        \dot{x}_D=r_2(x_A,x_B)-r_4(x_D).\\
    \end{cases}
\end{equation}
We endow the system with Michaelis--Menten kinetics and we find
multistationarity for a proper choice of kinetic constants. Such
multistationarity appears in the form of two equilibria, one stable $E_s$
and one unstable $E_u$. The dynamical intuition behind this is that the
system admits a saddle-node bifurcation \cite{V23}, implying a parameter
area where an unstable equilibrium coexists with a stable equilibrium. We
then simply consider initial conditions close to the unstable equilibrium
$E_u$. On the heteroclinic orbit connecting the two equilibria, we see
superlinear growth of the concentration $x_A$. More specifically, we
consider:
\begin{equation}\label{ex1mm}
    \begin{cases}
        \dot{x}_A=3-3\frac{x_Ax_B}{1+2x_B}-2x_A^2x_B,\\
        \dot{x}_B=2-3\frac{x_Ax_B}{1+2x_B}-x_A^2x_B,\\
        \dot{x}_C=3\frac{x_Ax_B}{1+2x_B}-x_C,\\
        \dot{x}_D=x_A^2x_B-x_D.\\
    \end{cases}
\end{equation}
$E_s:=(x_A,x_B,x_C,x_D)=(2,0.25,1,1)$ is a stable equilibrium while
$E_u:=(x_A,x_B,x_C,x_D)=(1,1,1,1)$ is unstable. With initial conditions
$x(0)=(1.01,1,1,1)$ we get the plot in Figure
\ref{fig_one}, where we find superlinear growth in the concentration of
$x_A$.

\begin{figure}
  \centering \includegraphics[width=0.7\textwidth]{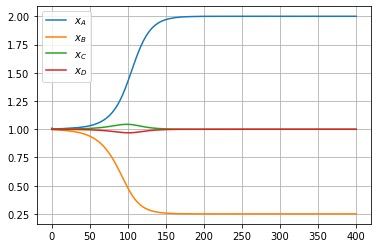}
  \caption{The system in \eqref{ex1mm} possesses an unstable equilibrium
    $E_u$ at $x=(1,1,1,1)$. In the figure, nearby initial conditions
    $x(0)=(1.01,1,1,1)$ converge to the stable equilibrium $E_s=(2, 0.25,
    1, 1).$ The trajectory of $x_A$ has a sigmoid shape with alternating
    regimes of superlinear, linear, and sublinear growth.}
    \label{fig_one}
    \end{figure}

The unstable-positive feedback in
\eqref{twinable1SM} admits an autocatalytic twin
\begin{equation}\label{twinableaut1}
\begin{pmatrix}
-1 & 2\\
1 & -1
\end{pmatrix}.
\end{equation}
With similar intuition, we can construct an example with superlinear
growth, based on \eqref{twinableaut1}, instead of
\eqref{twinable1SM}. Consider two species $A$ and $B$, and the following
four reactions.
\begin{align*}
    A &\underset{1}{\rightarrow}  \\
    A &\underset{2}{\rightarrow} B \underset{3}{\rightarrow} 2A \\
    B &\underset{4}{\rightarrow}
\end{align*}
Endowing the system with Michaelis--Menten kinetics, the following choice
of constants
\begin{equation}\label{ex1b}
    \begin{cases}
\dot{x}_A=-\frac{12x_A}{1+x_A}+8x_B-2x_A,\\
\dot{x}_B=\frac{12x_A}{1+x_A}-4x_B-\frac{8x_B}{1+3x_B},
    \end{cases}
\end{equation}
shows    multistationarity   with    superlinear    growth,   see    Figure
\ref{fig_1b}.  As  already pointed  out,  the  dynamics  of the  system  in
\eqref{ex1mm} and in \eqref{ex1b} is  indeed different. This is no surprise
since  the networks  must be  different for  stoichiometric reasons.  E.g.,
enlarging  the  stoichiometry  of the  non-autocatalytic  unstable-core  in
\eqref{twinable1SM}   to  a   network  that   admits  equilibria   requires
inflow-reactions, while for the  autocatalytic core in \eqref{twinableaut1}
the same construction requires outflow-reactions.

 \begin{figure}
 \centering
 \includegraphics[width=0.7\textwidth]{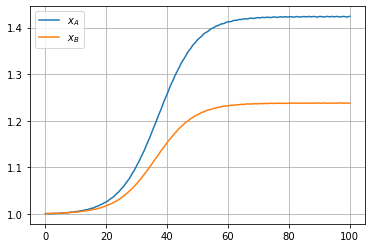}
   \caption{The system in \eqref{ex1b} possesses an unstable equilibrium
     $(x_A,x_B,x_C,x_D,x_E)=(1,1,1,1,1)$. In the figure, nearby initial
     conditions $x(0)=(1.01,1,1,1,1)$ show superlinear growth and
     convergence to a limit stable equilibrium.}
   \label{fig_1b}
\end{figure}

\subsection*{EXAMPLE B: Oscillations with unstable-positive feedback}
\label{subsec:B}
Consider five species $A,B,C,D,E$, and the following eight reactions.\\
\begin{minipage}{.6\textwidth}
\begin{align*}       
    A+B &\underset{1}{\rightarrow} C+E\\ B+C &\underset{2}{\rightarrow}
    E\\ C+D &\underset{3}{\rightarrow} A+E\\ 2B+D
    &\underset{4}{\rightarrow} E\\
\end{align*}
\end{minipage}%
\begin{minipage}{.2\textwidth}
\begin{align*}
       E &\underset{5}{\rightarrow}\\
    &\underset{F_A}{\rightarrow} A\\
&\underset{F_B}{\rightarrow} B\\
  &\underset{F_D}{\rightarrow} D\\
\end{align*}
\end{minipage}\\
The stoichiometric matrix reads
\begin{equation}
    S=
\begin{pmatrix}
    -1 & 0 & 1 & 0 & 0 & 1 & 0 & 0\\
    -1 & -1 & 0 & -2 & 0 & 0 & 1 & 0\\
    1 & -1 & -1 & 0 & 0 & 0 & 0 & 0\\
    0 & 0 & -1 & -1 & 0 & 0 & 0 & 1\\
    1 & 1 & 1 & 1 & -1 & 0 & 0 & 0\\
\end{pmatrix}.
\end{equation}
Consider the $3$-CS $\pmb{\kappa}$ defined as 
\begin{equation}
    \pmb{\kappa}=(\{A,B,C
    \},\{1,2,3\},\{J(A)=1,J(B)=2,J(C)=3\})
\end{equation} 
The network possesses a unique unstable-core,
\begin{equation} \label{pfnaSM}
S[\pmb{\kappa}]=
\begin{pmatrix}
   -1 & 0 & 1\\
    -1 & -1 & 0\\
    1 & -1 & -1   
\end{pmatrix},
\end{equation}
which is a non-autocatalytic unstable-positive feedback. In fact,
$\det S[\pmb{\kappa}]=1=(-1)^{3-1}$, and no principal
submatrix is unstable. The uniqueness can be checked via an elementary
computation, omitted here for brevity.  Note also that $S[\pmb{\kappa}]$
has a twin autocatalytic core $A$ of the form
\begin{equation}\label{autocattwin2}
A=
\begin{pmatrix}
   -1 & 0 & 1\\
    1 & -1 & 0\\
    1 & 1 & -1   
\end{pmatrix}.
\end{equation}
The associated dynamical system of the concentrations is:
\begin{equation}\label{posfeedback}
    \begin{cases}
        \dot{x}_A=F_A-r_1(x_A,x_B)+r_3(x_C,x_D),\\
        \dot{x}_B=F_B-r_1(x_A,x_B)-r_2(x_B,x_C)-2r_4(x_B,x_D),\\
        \dot{x}_C=r_1(x_A,x_B)-r_2(x_B,x_C)-r_3(x_C,x_D),\\
        \dot{x}_D=F_D-r_3(x_C,x_D)-r_4(x_B,x_D)\\
\dot{x}_E=r_1(x_A,x_B)+r_2(x_B,x_C)+r_3(x_C,x_D)+r_4(x_B,x_D)-r_5(x_E).
    \end{cases}
\end{equation}
We endow the system with Michaelis--Menten kinetics and we find oscillations for a certain choice of constants. For example, the following choice
\begin{equation}\label{ex2mm}
    \begin{cases}
        \dot{x}_A=1-8\frac{x_Ax_B}{(1+x_A)(1+x_B)}+91\frac{x_Cx_D}{1+90x_D},\\
        \dot{x}_B=5-8\frac{x_Ax_B}{(1+x_A)(1+x_B)}-91\frac{x_Bx_C}{1+90x_C}-20\frac{x_B^2x_D}{(1+x_B)^2(1+1.5x_D)},\\
        \dot{x}_C=8\frac{x_Ax_B}{(1+x_A)(1+x_B)}-91\frac{x_Bx_C}{1+90x_C}-91\frac{x_Cx_D}{1+90x_D},\\
        \dot{x}_D=2-91\frac{x_Cx_D}{1+90x_D}-10\frac{x_B^2x_D}{(1+x_B)^2(1+1.5x_D)},\\
    \dot{x}_E=8\frac{x_Ax_B}{(1+x_A)(1+x_B)}+91\frac{x_Bx_C}{1+90x_C}+91\frac{x_Cx_D}{1+90x_D}+10\frac{x_B^2x_D}{(1+x_B)^2(1+1.5x_D)}-5x_E,
    \end{cases}
\end{equation}
shows an unstable equilibrium for
$x=(1,1,1,1,1)$. Nearby initial conditions
$x(0)=(1.01,1,1,1,1)$ converge to a stable
periodic orbit. See Figure \ref{fig_two1} and Figure \ref{fig_two2}.

\begin{figure}
  \centering
  \includegraphics[width=\textwidth]{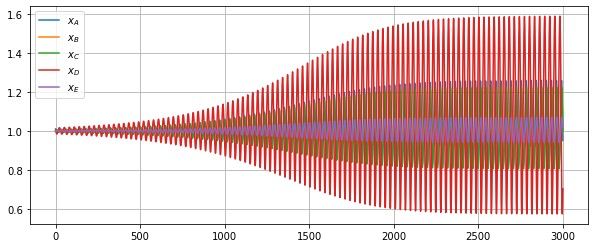}
  \caption{The system in \eqref{ex2mm} possesses an unstable equilibrium
    $(x_A,x_B,x_C,x_D,x_E)=(1,1,1,1,1)$. In the figure, nearby initial
    conditions $x(0)=(1.01,1,1,1,1)$ shows convergence to a limit cycle.}
    \label{fig_two1}
    \end{figure}

 \begin{figure}
  \centering
\includegraphics[width=\textwidth]{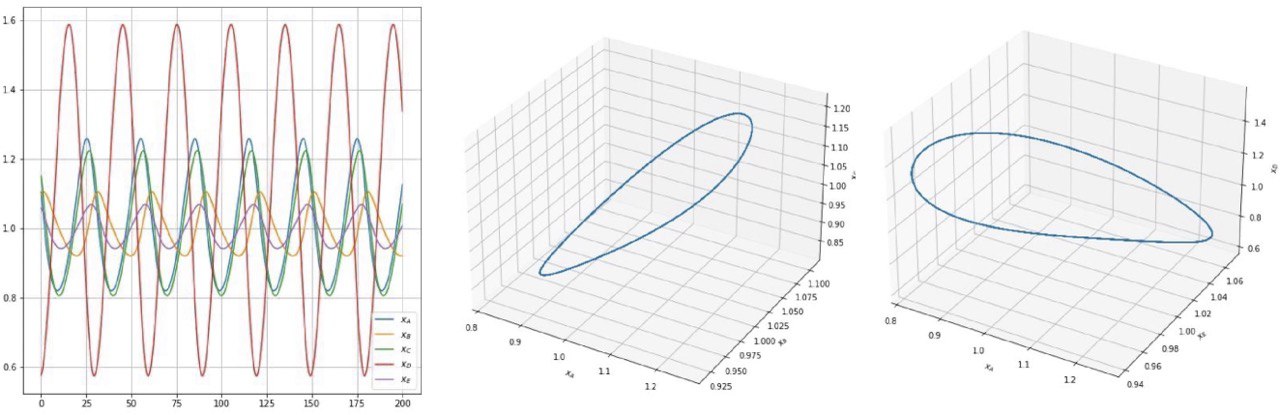}
  \caption{The picture shows numerical simulations for the system in
    \eqref{ex2mm}. The initial condition $x(0)=[1.10433115, 1.0969183,
      1.15016837, 0.57719267, 1.05835769]$ is chosen in proximity of the
    periodic orbit. The plot on the left shows the time-evolution of the
    concentrations $x(t)$, while the plots in the center and on the right
    depict in 3d the periodic orbit for $(x_A,x_B,x_C)$ and
    $(x_A,x_D,x_E)$, respectively.}
  \label{fig_two2}
\end{figure}

The designing idea behind the example is that any Child-Selection matrix
that is both Hurwitz-stable and $D$-unstable allows the detection of purely
imaginary eigenvalues of the Jacobian matrix of networks
with parameter-rich kinetic models, see \cite{VasHunt}. Specifically, this example enlarges the
matrix in \eqref{pfnaSM} to a $4\times4$ Hurwitz-stable stoichiometric
matrix:
\begin{equation}
B=
\begin{pmatrix}
-1 & 0 & 1 & 0\\
-1 & -1 &0 & -2\\
1 & -1 & -1 & 0\\
0 & 0 & -1 & -1
\end{pmatrix},
\end{equation}
which indeed possesses four negative-real-part eigenvalues. However, it is clear that the matrix
\begin{equation}\label{unstableps}
B(\varepsilon)=
\begin{pmatrix}
-1 & 0 & 1 & 0\\
-1 & -1 &0 & -2\varepsilon\\
1 & -1 & -1 & 0\\
0 & 0 & -1 & -\varepsilon
\end{pmatrix}
\end{equation}
loses stability with $\epsilon \rightarrow 0$, in a similar fashion as the
arguments in Lemma~\ref{lemma:CSinst} of the main text (Lemma~\ref{lemma:CSinstSI} here in the SM). Moreover, the matrix
$B(\varepsilon)$ is invertible for any $\varepsilon$ (multiplying a column
just multiplies the determinant), and thus the only possibility for the
loss of stability is a crossing of a pair of purely imaginary
eigenvalues. The full example inherits this core feature via a proper
choice of Michaelis-Menten parameters. Likely this results in a Hopf
bifurcation \cite{Hsubook} and luckily it results in a supercritical Hopf
bifurcation, generating a stable periodic orbit.

An oscillatory example based on the autocatalytic twin in
\eqref{autocattwin2} has been presented in \cite{VasHunt}, and we omit here
the full details. We just report the example for completeness, made of 5
species $A,B,C,D,E$, and 8 reactions.\\
\begin{minipage}{.6\textwidth}
\begin{align*}       
        A &\underset{1}{\rightarrow} B+C\\
        \quad\quad\;\;B &\underset{2}{\rightarrow} C\\
\quad\quad\;\;C+D&\underset{3}{\rightarrow} A\\
C &\underset{4}{\rightarrow} E     
\end{align*}
\end{minipage}%
\begin{minipage}{.1\textwidth}
\begin{align*}
  \quad\quad\;\;D
  &\underset{5}{\rightarrow} 2B\\
\quad\quad\;\;\quad D+E&\underset{6}{\rightarrow}2E\\
E&\underset{7}{\rightarrow}\\
&\underset{F_D}{\rightarrow} D
\end{align*}
\end{minipage}\\
The following Michaelis--Menten system has an unstable
equilibrium at $x=(1,1,1,1,1)$, with nearby initial
conditions converging to a stable periodic orbit.
\begin{equation}
    \begin{cases}
\dot{x}_A=-2x_A+8\frac{x_Cx_D}{1+3x_D},\\
    \dot{x}_B=2x_A-8\frac{x_B}{1+x_B}+4 \frac{x_D}{1+x_D},\\
    \dot{x}_C=2x_A+8\frac{x_B}{1+x_B}-8\frac{x_Cx_D}{1+3x_D}-64 \frac{x_C}{1+15x_C},\\
    \dot{x}_D=-8\frac{x_Cx_D}{1+3x_D}-2 \frac{x_D}{1+x_D}-512 \frac{x_D}{1+63x_D} \frac{x_E}{1+3x_E}+5,\\
    \dot{x}_E=64 \frac{x_C}{1+15x_C}+512 \frac{x_D}{1+63x_D} \frac{x_E}{1+3x_E}-72 \frac{x_E}{1+11x_E}.\\
    \end{cases}
\end{equation}

\subsection*{EXAMPLE C: Oscillations with unstable-negative feedback and
  mass-action variation}\label{subsec:C}
Consider six species $A,B,C,D,E,F$ and the following eleven reactions.\\
\begin{minipage}{0.5\textwidth}
\begin{align*}       
    A+B &\underset{1}{\rightarrow} A+F\\
    B+C &\underset{2}{\rightarrow} B+F\\
    C+D &\underset{3}{\rightarrow} C+F\\
    D+E &\underset{4}{\rightarrow} D+F\\
    E+A &\underset{5}{\rightarrow} E+F\\
    F&\underset{6}{\rightarrow}
\end{align*}
\end{minipage}%
\begin{minipage}{.2\textwidth}
\begin{align*}
   &\underset{F_A}{\rightarrow} A\\
&\underset{F_B}{\rightarrow} B\\
  &\underset{F_C}{\rightarrow} C\\
  &\underset{F_D}{\rightarrow} D\\
 &\underset{F_E}{\rightarrow} E\\
\end{align*}
\end{minipage}\\
The stoichiometric matrix $S$ is
\begin{equation}\label{ex3S}
    S=\begin{pmatrix}
   0 & 0 & 0 & 0 & -1& 0 & 1 & 0 & 0 & 0 & 0\\
  -1 & 0 & 0 & 0 & 0 & 0 & 0 & 1 & 0 & 0 & 0\\
    0 & -1 & 0 & 0 & 0 & 0 & 0 & 0 & 1 & 0 & 0\\
    0 & 0 & -1 & 0 & 0& 0 & 0 & 0 & 0 & 1 & 0\\
    0 & 0 & 0 & -1 & 0 & 0 & 0 & 0 & 0 & 0 & 1\\
    1 & 1 & 1 & 1 & 1 & -1 & 0 & 0 & 0 & 0 & 0\\
\end{pmatrix}.
\end{equation}
Reactions 1, 2, 3, 4, and 5 have been chosen explicitly catalytic to reduce
as much as possible the number of reactions and species. Clearly, the
network contains only the unstable-core
\begin{equation}\label{negfeedex3}
    S_{u}=
    \begin{pmatrix}
    0 & 0 & 0 & 0 & -1\\
    -1 & 0 & 0 & 0 & 0\\
    0 & -1 & 0 & 0 & 0\\
    0 & 0 & -1 & 0 & 0\\
    0 & 0 & 0 & -1 & 0
    \end{pmatrix},
\end{equation}
which constitutes an unstable-negative feedback and thus
not-autocatalytic. In fact, $S_{u}$ has a pair of complex-conjugated
unstable eigenvalues, and all its strict principal submatrices have only
zero eigenvalues; thus they are not Hurwitz-unstable. The dynamical systems
of the concentrations is
\begin{equation}\label{ds}
    \begin{cases}
        \dot{x}_A=F_A-r_5(x_A,x_E),\\
        \dot{x}_B=F_B-r_1(x_A,x_B),\\
        \dot{x}_C=F_C-r_2(x_B,x_C),\\
        \dot{x}_D=F_D-r_3(x_C,x_D),\\
        \dot{x}_E=F_E-r_4(x_D,x_E),\\ 
        \dot{x}_F=r_1(x_A,x_B)+r_2(x_B,x_C)+r_3(x_C,x_D)+r_4(x_D,x_E)+r_5(x_A,x_E)-r_6(x_F).
    \end{cases}
\end{equation}
The system in \eqref{ds} admits oscillations when endowed with
Michaelis--Menten kinetics. For example, for the choice of parameters
\begin{equation}\label{ex3mm}
    \begin{cases}
        \dot{x}_A=1-2\frac{x_Ex_A}{1+x_A},\\
        \dot{x}_B=1-2\frac{x_Ax_B}{1+x_B},\\
        \dot{x}_C=1-2\frac{x_Bx_C}{1+x_C},\\
        \dot{x}_D=1-2\frac{x_Cx_D}{1+x_D},\\
        \dot{x}_E=1-2\frac{x_Dx_E}{1+x_E},\\ 
        \dot{x}_F=2\frac{x_Ex_A}{1+x_A}+2\frac{x_Ax_B}{1+x_B}+2\frac{x_Bx_C}{1+x_C}+2\frac{x_Cx_D}{1+x_D}+2\frac{x_Dx_E}{1+x_E}-5x_F,
    \end{cases}
\end{equation}
the point $x=(1,1,1,1,1,1)$ is an unstable
equilibrium. Nearby initial conditions $x(0)=(1,1,1,1,1.01,1)$ converge to
a stable periodic orbit. See Figure \ref{fig_three1} and Figure \ref{fig_three2}.

\begin{figure}
    \centering \includegraphics[width=0.9\textwidth]{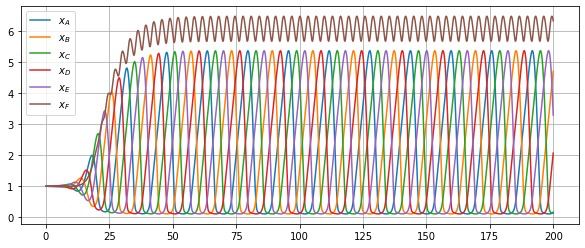}
    \caption{The system in \eqref{ex3mm} possesses an unstable equilibrium
      at $x=(1,1,1,1,1,1)$. In the figure, nearby
      initial conditions $x(0)=(1,1,1,1,1.01,1)$ shows convergence to a
      limit cycle.}
      \label{fig_three1}
      \end{figure}

\begin{figure}
    \centering
\includegraphics[width=0.9\textwidth]{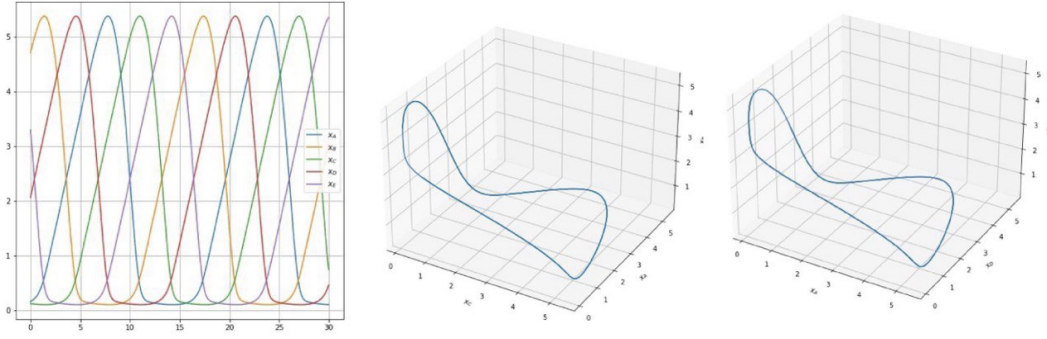}
    \caption{The picture shows numerical simulations for the system in
      \eqref{ex3mm}. The initial condition $x(0)=[0.15944159, 4.70616422,
        0.12197977, 2.0658755, 3.29519854, 6.33865963]$ is chosen in
      proximity of the periodic orbit. The plot on the left shows the
      time-evolution of the concentrations of the species $A,B,C,D,E$ in
      the negative feedback, while the plots in the center and on the right
      depict in 3d the periodic orbit for $(x_A,x_B,x_C)$ and
      $(x_A,x_D,x_E)$, respectively.}
    \label{fig_three2}
\end{figure}

An identical construction as above can be also used in a mass-action
context. We just substitute the catalysts $(A,B,C,D,E)$ with
$(2A,2B,2C,2D,2E)$. This does not change the stoichiometric matrix of the
network, but it changes the Jacobian to our advantage:\\
\begin{minipage}{.6\textwidth}
\begin{align*}       
  2A+B &\underset{1}{\rightarrow} 2A+F\\
  2B+C &\underset{2}{\rightarrow} 2B+F\\
  2C+D &\underset{3}{\rightarrow} 2C+F\\
  2D+E &\underset{4}{\rightarrow} 2D+F\\
  2E+A &\underset{5}{\rightarrow} 2E+F\\
  F&\underset{6}{\rightarrow}
\end{align*}
\end{minipage}%
\begin{minipage}{.1\textwidth}
\begin{align*}
   &\underset{F_A}{\rightarrow} A\\
&\underset{F_B}{\rightarrow} B\\
  &\underset{F_C}{\rightarrow} C\\
  &\underset{F_D}{\rightarrow} D\\
 &\underset{F_E}{\rightarrow} E\\
\end{align*}
\end{minipage}\\
It is easy to check that we get indeed the same stoichiometric matrix as
\eqref{ex3S}. Again, there is an unstable-negative feedback of the form in
\eqref{negfeedex3}. We underline that in a mass-action setting the presence
of an unstable core does not automatically guarantee that there exists an
unstable equilibrium, even if it is the case in this example. The
mass-action system reads
\begin{equation}\label{ex4ma}
    \begin{cases}
        \dot{x}_A=F_A-k_5x_Ax_E^2,\\
        \dot{x}_B=F_B-k_1x_A^2x_B,\\
        \dot{x}_C=F_C-k_2x_B^2x_C,\\
        \dot{x}_D=F_D-k_3x_C^2x_D,\\
        \dot{x}_E=F_E-k_4x_D^2x_E,\\ 
        \dot{x}_F=k_5x_Ax_E^2+k_1x_A^2x_B+k_2x_B^2x_C+k_3x_C^2x_D+k_4x_D^2x_E-k_6x_F.
    \end{cases}
\end{equation}
For the choice of mass-action constants 
\begin{equation}\label{mareaction}
   (k_1,k_2,k_3,k_4,k_5,k_6,F_A,F_B,F_C,F_D,F_E)=
(1,1,1,1,1,5,1,1,1,1,1),
\end{equation} 
the equilibrium $(x_A,x_B,x_C,x_D,x_E,x_F)=(1,1,1,1,1,1)$ is unstable and
nearby trajectory converges to a stable periodic orbit. For the initial
conditions $(x_A,x_B,x_C,x_D,x_E,x_F)=(1,1,1,1,1.01,1)$ the simulation
shows convergence to a limit cycle. See Figure \ref{figfour1} and \ref{fig_four2}.

\begin{figure}
  \centering
\includegraphics[width=\textwidth]{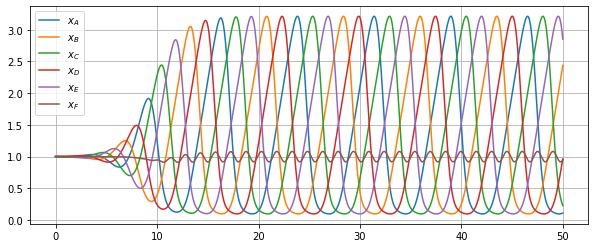}
  \caption{The system in \eqref{ex4ma} with reaction rates
    \eqref{mareaction} possesses an unstable equilibrium at
    $x=(1,1,1,1,1,1)$. In the figure, nearby initial conditions
    $x(0)=(1,1,1,1,1.01,1)$ convergence to a limit cycle.}
    \label{figfour1}
    \end{figure}
    
 \begin{figure}{
    \centering
\includegraphics[width=\textwidth]{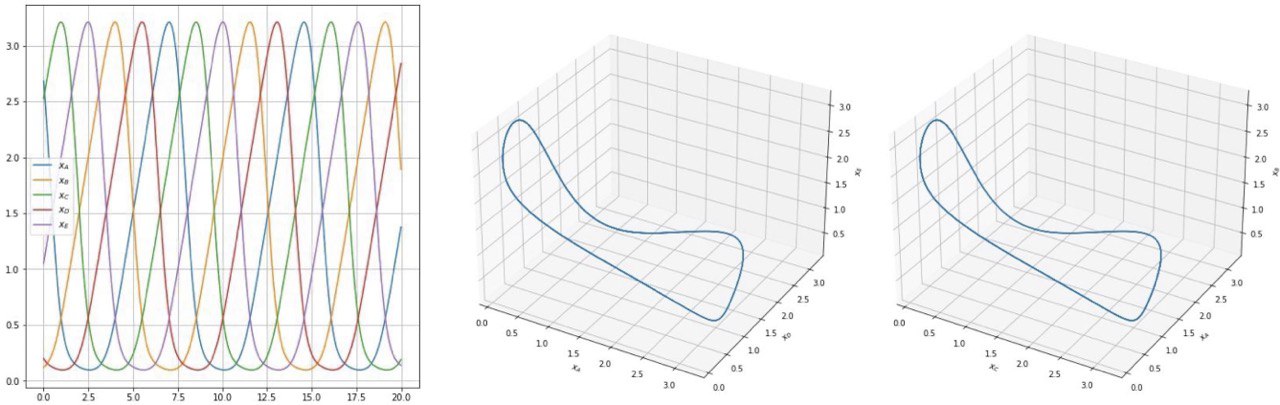}
  \caption{The picture shows numerical simulations for the system in
    \eqref{ex4ma} with reaction rates \eqref{mareaction}. The initial condition $x(0)=[2.68177111, 0.12190425,
      2.52781617, 0.20023526, 1.05389654, 0.97774366]$ is chosen in
    proximity of the periodic orbit. The plot on the left shows the
    time-evolution of the concentrations of the species $A,B,C,D,E$ in the
    negative feedback, while the plots in the center and on the right
    depict in 3d the periodic orbit for $(x_A,x_B,x_C)$ and
    $(x_A,x_D,x_E)$, respectively.}
    \label{fig_four2}}
\end{figure}

\subsection*{EXAMPLE D: Non-autocatalytic unstable-positive
  feedbacks in the dual futile cycle}\label{subsec:D}
The dual futile cycle
\begin{align} \label{SI_doublephosphorilation}
\begin{cases}
A+E_1 \overset{1}{\underset{2}{\rightleftharpoons}} I_1 \underset{3}{\longrightarrow} B+E_1 \overset{7}{\underset{8}{\rightleftharpoons}} I_3 \underset{9}{\longrightarrow} C+E_1;\\
A+E_2 \underset{6}{\longleftarrow} I_2 \overset{5}{\underset{4}{\rightleftharpoons}} B+E_2 \underset{12}{\longleftarrow} I_4 \overset{11}{\underset{10}{\rightleftharpoons}} C+E_2 .
\end{cases}
\end{align}
is a relevant example of a phosphorylation system; it is constituted by two
blocks of distributive and sequential phosphorylation steps, where $A$,
$B$, $C$ are the core molecules, $E_1$ and $E_2$ are the enzymes, and
$I_1$, $I_2$, $I_3$, $I_4$ are intermediates. Phosphorylation is widespread
in biochemical processes \cite{Cohen:1989}, and it is ubiquitous in
cell-signaling in general and specifically in the MAPK cascade
\cite{HuangFerrell:1996}. Such networks can also be seen as a special case
of post-translational modification systems
\cite{ThomsonGunawardena:JTB:2009}.  Hence it is no surprise that they have
attracted great interest from the mathematical community, and we do not aim
here at an exhaustive literature review. A general framework of analysis
for such type of networks has been addressed in
\cite{ThomsonGunawardena:09}, and more recently lead to the definition of
MESSI networks \cite{MESSI:18}.

Specifically the dual futile cycle, moreover, has posed quite challenging
problems as a mathematical toy model.  When endowed with Mass Action
kinetics, the dual futile cycle has been analyzed via a singular
perturbation to a \emph{monotone system} \cite{WangSontag:2008}. It was
proven that the network admits a \emph{toric equilibrium}
\cite{AliciaetalToric:2012}. It was also shown that the system admits a
\emph{cusp bifurcation} and consequent bistability
\cite{HellRendall:2015}. Next, the region of parameters that admits
multistationarity has been fully described \cite{Feliu:2020}. In contrast,
it is still an open question whether or not the system admits a Hopf
bifurcation \cite{CarstenHopfExclusion19}, with some numerical evidence
that such bifurcation cannot occur, but no analytical or computer-algebra
proof.

The stoichiometric matrix reads
\begin{equation}
S=
\begin{blockarray}{ccccccccccccc}
& 1 & 2 & 3 & 4 & 5 & 6 & 7 & 8 & 9 & 10 & 11 & 12\\
\begin{block}{c(cccccccccccc)}
A & -1 & 1 & 0 & 0 & 0 & 1 & 0 & 0 & 0 & 0 & 0 & 0\\
B & 0 & 0 & 1 & -1 & 1 & 0 & -1 & 1 & 0 & 0 & 0 & 1\\
C & 0 & 0 & 0 & 0 & 0 & 0 & 0 & 0 & 1 & -1 & 1 & 0\\
E_1 & -1 & 1 & 1 & 0 & 0 & 0 & -1 & 1 & 1 & 0 & 0 & 0\\
E_2 & 0 & 0 & 0 & -1 & 1 & 1 & 0 & 0 & 0 & -1 & 1 & 1\\
I_1 & 1 & -1 & -1 & 0 & 0 & 0 & 0 & 0 & 0 & 0 & 0 & 0\\
I_2 & 0 & 0 & 0 & 1 & -1 & -1 & 0 & 0 & 0 & 0 & 0 & 0\\
I_3 & 0 & 0 & 0 & 0 & 0 & 0 & 1 & -1 & -1 & 0 & 0 & 0\\
I_4 & 0 & 0 & 0 & 0 & 0 & 0 & 0 & 0 & 0 & 1 & -1 & -1\\
\end{block}
\end{blockarray}\;.
\end{equation}

We firstly note an obvious $\mathbb{Z}_2$-symmetry $\sigma$ in the
network. For species,
\begin{equation}\sigma(A,B,C,E_1,E_2,I_1,I_2,I_3,I_4)=(C,B,A,E_2,E_1,I_4,I_3,I_2,I_1),
\end{equation} and for reactions 
\begin{equation}
\sigma(1,2,3,4,5,6,7,8,9,10,11,12)=(10,11,12,7,8,9,4,5,6,1,2,3).
\end{equation}

The Jacobian matrix $G$ of the associated dynamical system of the
concentrations is a $9\times9$ matrix. A direct computer-algebra
computation of the characteristic polynomial
\begin{equation}
g(\lambda)=\sum_{k=0}^9(-1)^kc_k\lambda^{9-k}
\end{equation}
shows that $c_7=c_8=c_9\equiv 0$, implying that $G$ is at most of rank
$\rk(G)=6$. In particular, there are three linear conserved quantities:
\begin{equation}
\begin{cases}
w_1=x_{E_1}+x_{I_1}+x_{I_3};\\
w_2=x_{E_2}+x_{I_4}+x_{I_2};\\
w_3=x_{A}+x_{B}+x_{C}+x_{I_1}+x_{I_2}+x_{I_3}+x_{I_4},
\end{cases}
\end{equation}
that correspond to three independent left kernel vectors of the
stoichiometric matrix. Note that $\sigma(w_1)=w_2$ and $\sigma(w_3)=w_3$.
We expand the remaining coefficients in elementary CB-components as in
\eqref{eq:completeCBexpansionSI}. The expansion for the first four
coefficients surprisingly shows that $\det{S[\pmb{\kappa}]}=(-1)^k$ or
$\det{S[\pmb{\kappa}]}=0$ for all $k$-CS $\pmb{\kappa}$, $k=1,2,3,4.$ This
implies that $\sign(c_1)=\sign(c_2)=\sign(c_3)=\sign(c_4)=1$ for all
choices of parameters, and that in particular there are no
generalized-unstable-positive feedbacks (and thus no unstable-positive
feedbacks) with less than 5 species. The same expansion for $c_5$ shows
that $|\det S[\pmb{\kappa}]|=1$ or $\det S[\pmb{\kappa}]=0$ for all
$5$-CS. Moreover, there are only two generalized-positive-feedbacks with
$\det S[\pmb{\kappa}]=(-1)^{5-1}$:
\begin{equation}
S_{1}=
\begin{blockarray}{cccccc}
& 1 & 7 & 9 & 4 & 6\\
\begin{block}{c(ccccc)}
\\
A & -1 & 0 & 0 & 0 & 1\\
E_1 & -1 & -1 & 1 & 0& 0\\
I_3 & 0 & 1 & -1 & 0 & 0\\
B & 0 & -1 & 0 & -1 & 0\\
I_2 & 0 & 0 & 0 & 1 & -1\\
\\
\end{block}
\end{blockarray}
\quad\text{ and }\quad S_2=
\begin{blockarray}{cccccc}
& 10 & 4 & 6 & 7 & 9\\
\begin{block}{c(ccccc)}
\\
C & -1 & 0 & 0 & 0 & 1\\ 
E_2 & -1 & -1 & 1 & 0& 0\\
I_2 & 0 & 1 & -1 & 0 & 0\\
B & 0 & -1 & 0 & -1 & 0\\
E_3 & 0 & 0 & 0 & 1 & -1\\
\\
\end{block}
\end{blockarray}\;,
\end{equation}
with $\det S_1=\det S_2 =1$.  We note again that $\sigma(S_1)=S_2$, which
is reflected in the fact that the matrices are indeed identical. The same
expansion for the coefficient $c_6$ again shows $|\det S[\pmb{\kappa}]|=1$
or $\det S[\pmb{\kappa}]=0$, for all $6$-CS, which concludes that $|\det
S[\pmb{\kappa}]|=1$ for all invertible Child-Selection matrices in the
network. Moreover, there are eleven 6-Child-Selection matrices
$S[\pmb{\kappa}]$ with $\det S[\pmb{\kappa}]=(-1)^{6-1}=-1$. Among those
eleven, five matrices contain as a proper principal submatrix $S_1$, and
five contain $S_2$ (again related by the symmetry $\sigma$) and thus these
are not unstable-cores. The last remaining Child-Selection matrix
identifies the following generalized-unstable-positive feedback:
\begin{equation}
S_3=
\begin{blockarray}{ccccccc}
& 1 & 7 & 9 & 10 & 4 & 6\\
\begin{block}{c(cccccc)}
\\
A & -1 & 0 & 0 & 0 & 0 & 1\\
E_1 & -1 & -1 & 1 & 0& 0 & 0\\
I_3 & 0 & 1 & -1 & 0 & 0 & 0\\
C & 0 & 0 & 1 & -1 & 0 & 0\\
E_2 & 0 & 0 & 0 & -1 & -1 & 1\\
I_2& 0 & 0 & 0 & 0 & 1 & -1\\
\\
\end{block}
\end{blockarray}\;,
\end{equation}
with $\det S_3 = (-1)^{6-1}=-1$, and that satisfies $\sigma(S_3)=S_3$.  We
then note that $S_1, S_2, S_3$ are not Metzler matrices, and thus they are
non-autocatalytic. However, they all possess autocatalytic twins of the
following form:
\begin{equation}
    A_1=A_2= \begin{pmatrix}
      -1 & 0 & 0 & 0 & 1\\ 
      1 & -1 & 1 & 0& 0\\
      0 & 1 & -1 & 0 & 0\\
      0 & 1 & 0 & -1 & 0\\
      0 & 0 & 0 & 1 & -1
    \end{pmatrix}\quad\text{ and }
    A_3=\begin{pmatrix}
    -1 & 0 & 0 & 0 & 0 & 1\\
    1 & -1 & 1 & 0& 0 & 0\\
    0 & 1 & -1 & 0 & 0 & 0\\
    0 & 0 & 1 & -1 & 0 & 0\\
    0 & 0 & 0 & 1 & -1 & 1\\
    0 & 0 & 0 & 0 & 1 & -1
    \end{pmatrix},
\end{equation}
which, via Cor.~\ref{cor:twingenfeed}, implies that $S_1$, $S_2$, and
$S_3$ are unstable-positive feedbacks. Finally, we conjecture that the
above features hold for any $n$-futile cycles, built iteratively by adding
further distributive and sequential phosphorylation steps. That is, we
conjecture that only unstable-positive feedbacks with matrices as $S_1=S_2$
or $S_3$ appear and that $|\det S[\pmb{\kappa}]|=1$ for all invertible
Child-Selection matrices. We hope that such a striking structure may help a
further understanding of this remarkable family of networks.

\section*{Unstable-positive feedback with unstable dimension greater than one}\label{sec:SIupfex}
Consider the following
unstable-positive feedback:
\begin{equation}\label{examplecore3}
S^+_1:=\begin{pmatrix}
    -1 & 0 & 0 & 0 & 4\\
    4 & -1 & 0 & 0 & 0\\
    0 & 4 & -1 & 0 & 0\\
    0 & 0 & 4 & -1 & 0 \\
    0 & 0 & 0 & 4 & -1
\end{pmatrix}.
\end{equation}
It is straightforward to check that $S^+_1$ is indeed an unstable-positive
feedback using the same arguments as for $S^+$ in the main text, Section 6.  Moreover, $S^+_1$
has two eigenvalues with a negative-real part and three eigenvalues with a
positive-real part:
\begin{equation}
    (\lambda_1, \lambda_2, \lambda_3, \lambda_4, \lambda_5)\approx(-4.23607 \pm 2.35114 i, 0.236068 \pm 3.80423 i, 3).
\end{equation}

\bibliographystyle{vancouver}
\bibliography{references}

\end{document}